\documentclass[onecolumn,10pt]{IEEEtran}
\usepackage{graphicx}
\usepackage{epstopdf}
\usepackage{amsmath}
\usepackage{amsthm}
\usepackage{amsfonts}
\usepackage[justification=centering]{caption}
\usepackage{color}
\theoremstyle{}
\newtheorem{theorem}{Theorem}
\newtheorem{lemma}{Lemma}
\newtheorem{definition}{Definition}
\newtheorem{corollary}{Corollary}

\newtheorem{example}{Example}
\newtheorem{remark}{Remark}
\newtheorem{construction}{Construction}

\newtheorem{proposition}{Proposition}
\usepackage{bm}
\usepackage{booktabs}
\usepackage{multirow}
\usepackage{algorithm}
\usepackage{algpseudocode}
\usepackage{subfigure}
\usepackage{amssymb}
\usepackage{cite}
\usepackage{multicol}
\usepackage{lipsum}% dummy text
\usepackage{pstricks}

\newcommand{\tabcaption}{\def\@captype{table}\caption}

\begin{document}
\title{A Novel Transformation Approach of Shared-link Coded Caching Schemes for Multiaccess  Networks
\author{Minquan Cheng, Kai Wan,~\IEEEmembership{Member,~IEEE,} Dequan Liang,
Mingming Zhang,
and~Giuseppe Caire,~\IEEEmembership{Fellow,~IEEE}
\thanks{M. Cheng, D. Liang, and M. Zhang  are with Guangxi Key Lab of Multi-source Information Mining $\&$ Security, Guangxi Normal University,
Guilin 541004, China  (e-mail: chengqinshi@hotmail.com, dequan.liang@hotmail.com, ztw$\_$07@foxmail.com).}
\thanks{K. Wan and G. Caire are with the Electrical Engineering and Computer Science Department, Technische Universit\"{a}t Berlin,
10587 Berlin, Germany (e-mail: kai.wan@tu-berlin.de, caire@tu-berlin.de).  The work of K.~Wan and G.~Caire was partially funded by the European Research Council under the ERC Advanced Grant N. 789190, CARENET.}
}
}

\date{}
\maketitle

\begin{abstract}
This paper considers the multiaccess  coded caching systems formulated by Hachem et al., including a central server containing $N$ files connected to $K$ cache-less users through  an  error-free shared link, and $K$ cache-nodes, each equipped with a cache memory size of $M$ files. Each user has access to $L$ neighbouring cache-nodes with a cyclic wrap-around topology.
The coded caching scheme proposed by Hachem et al. suffers from the case that $L$ does not divide $K$, where the needed number of transmissions (a.k.a. load) is at most four times the load expression for the case where $L$ divides $K$.
  Our main contribution is to propose a novel {\it transformation} approach to smartly extend  the schemes satisfying some conditions for the well known shared-link caching systems to the multiaccess caching systems. Then we can get many coded caching schemes with different subpacketizations for multiaccess coded caching system. These resulting schemes have the maximum local caching gain (i.e., the cached contents stored at any $L$ neighbouring cache-nodes are different such that the number of retrieval packets by each user from the connected cache-nodes is maximal) and the same coded caching gain as the original schemes. Applying the transformation approach to the well-known  shared-link coded caching scheme proposed by Maddah-Ali and Niesen, we obtain a new multiaccess coded caching scheme that achieves the same load as the scheme of Hachem et al. but for any system parameters. Under the constraint of the cache placement used in this new multiaccess coded caching scheme, our delivery strategy is approximately optimal when $K$ is sufficiently large. Finally, we also show that the transmission load of the proposed scheme can be further reduced by compressing the multicast message.
\end{abstract}

\section{Introduction}
\IEEEPARstart{W}{ireless} networks are increasingly under stress because of the ever-increasing large number of wireless devices consuming high-quality multimedia content (e.g., on-demand video streaming).  The high temporal variability of network traffic results in congestions during the peak traffic times and underutilization of the network during off-peak traffic times. Caching can effectively shift traffic from peak to off-peak times \cite{BBD}, by storing fractions of popular content at users' local memories  during the peak traffic times, such that   users can be partly served from their local caches, thereby reducing the network traffic.

Coded caching was originally proposed by Maddah-Ali and Niesen (MN) in \cite{MN} for the shared-link broadcast network, where a single server with access to a library containing $N$ equal-length files  is connected to $K$ users through a shared link. Each of the $K$ users  has a cache memory which can store up to $M$ files. A coded caching scheme operates in two phases. In the placement phase, the server populates the users' caches without knowledge of future demands.  In the delivery phase, each user demands one file.  According to users' demands and caches, the server broadcasts  coded messages through the shared link to all users such that each user's demand is satisfied. The goal is minimize the worst-case number of transmissions  normalized by the file size (referred to as {\it worst-case load} or just {\it load}) in the delivery phase among all possible demands.
The MN coded caching scheme   uses a combinatorial design in the placement phase, such that in the delivery phase each message broadcasted by the server can simultaneously satisfy multiple users' demands.  The achieved load of the MN scheme is
\begin{align}
R_{\text{MN}}=\frac{K(1- M/N)}{ KM/N+1}, \ \forall M=\frac{N t}{K}: t\in \left\{0, 1,\ldots,  K  \right\}.  \label{eq:MN}
\end{align}
When $N \geq K$, the MN scheme was showed to be optimal under uncoded placement (i.e., each user directly copies some bits of files in its cache)~\cite{WTP,YMA} and to be generally order optimal within a factor of $2$~\cite{yufactor2TIT2018}.

A main drawback of the MN scheme is its high subpacketization. The first coded caching scheme with reduced subpacketization was   proposed in~\cite{SJTLD} with a grouping strategy.
The authors in \cite{YCTC} proposed a class of caching schemes based on the concept of placement delivery array (PDA), and showed that the MN scheme is   a special PDA, referred to as MN PDA.
Based on the concept of PDA, various caching schemes were proposed, e.g.,~\cite{YCTC,CJWY,CJTY,CJYT,MW,ZCJ,YTCC,SZG}, in order to reduce the subpacketization of the MN scheme.

\subsection{Multiaccess  caching}
 % In order to meet the business diversity needs of future terminals, the authors in \cite{RKEA} considered different types of overlapping networks for the first time to form a network and pointed out that future mobile information systems may be based on wireless networks in which multiple access networks coexist. The wireless networks are composed of densely deployed wireless access points and a small number of deployed cellular base stations. The wireless access points have a small coverage area and relatively large data rate, while the cellular base station has a large coverage area and a relatively large data rate small. Normally, an end user can connect to a few wireless access points nearby and receive signals from cellar base-stations nearby. Thereby, we can reduce the transmission load of cellar base-stations by placing caches at wireless access points such that the cached content can be accessed by the end users nearby.
\begin{figure}%[ht]
%\vspace{-2mm}
\centerline{\includegraphics[scale=0.8]{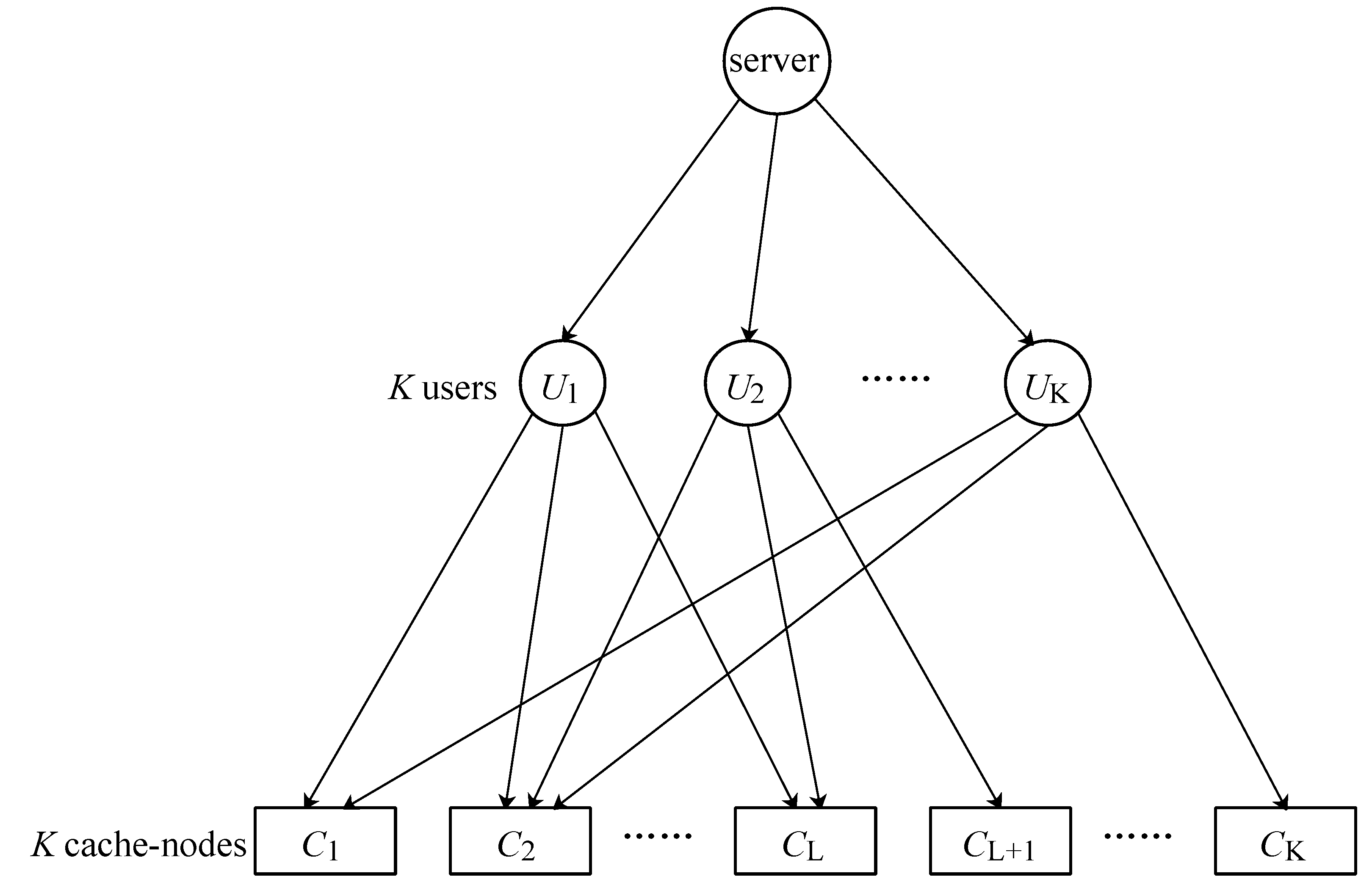}}
\caption{\small The $(K,L,M,N)$    multiaccess  coded caching system.}
\label{multi-access-system}
%\vspace{-5mm}
\end{figure}

Caching at the wireless edge nodes is a promising way to boost the
spatial and spectral  efficiency, for the sake of  reducing communication cost of wireless systems~\cite{FemtoCaching,LiuWirelesscaching}. Edge caches enable  storing Internet-based content, including web objects, videos and software updates.
 Compared to the end-user-caches, which are heavily limited by the storage size  of each user's device (e.g., a library size is usually much larger than the storage size of a mobile device) and only useful to  one user, the caches at edge devices such as local helpers, mobile edge
computing (MEC) servers and hotspots,  can be made much larger and can be accessed my multiple users.
 In this paper, we consider an ideal coded edge caching model with a line multiaccess topology, referred to as multiaccess  caching model (illustrated in Fig.~\ref{multi-access-system}), which was  originally introduced in   \cite{JHNS}.
 Different from  the shared-link caching system, in the   multiaccess  caching system, there are $K$ cache-nodes each of storage $M$ files, and $K$ cache-less users each of which can   access   $L$ cache-nodes in a cyclic wrap-around fashion, while also receiving the broadcast transmission directly from the server.
As assumed in   \cite{JHNS},   the   cache-nodes in the considered model are accessed at no load cost; that is, we count only the broadcasted load of the server.
This can be justified by the fact that the access to the  cache-nodes can be  offloaded  on a different network (e.g.,  a wifi offloading  for hotspots).

%Furthermore, edge cahing is a widely deployed solution that can improve better service for users while ensuring high reliability and high quality experience \cite{GBKO}.  The   multiaccess  caching problem can also be regarded as the idea model of edge caching. For example, the Content Delivery Network (CDN)   distributes the content storage in the central server to the edge of the network, such as the cache server and micro base station, and users can automatically connect to the cache server that is healthy and available and closest to the user. since the concept of mobile edge computing (MEC) was originally proposed in \cite{FAGA} \cite{JPJM}, in addition to the problem of end-to-end latency, how to migrate cloud computing tasks to the edge to reduce network load is another issue that MEC needs to consider \cite{BLLX}. Therefore, to a certain extent, the mobile edge computing model, which can effectively alleviate network congestion, can also be a possible extension of the multiaccess  coded caching problem.

A multiaccess  coded caching scheme was proposed in \cite{JHNS}, which deals with the following two cases separately.
\begin{itemize}
\item When $L$ divides  $K$ (i.e., $L|K$), the authors in \cite{JHNS} proposed   a scheme that divides  the $K$ users into $L$ groups, where the MN scheme is then used in  each user group. This results in a total load equal to
\begin{align}
R_{\text{HKD}}=\frac{K(1-LM/N)}{ KM/N+1}, \ \forall M=\frac{N t}{K}: t\in \left\{0, 1,\ldots, \left\lfloor \frac{K}{L}\right\rfloor \right\}.  \label{eq:R_HKD}
\end{align}
%Here $1-\frac{LM}{N}$ is the local caching gain, i.e., a portion  $\frac{LM}{N}$ of the bits in the desired file of each user can be directly retrieved from the connected caches of this user. This implies that any two cache-nodes which is connected to some common user(s) do not cache the same packets. $\frac{K M}{N}+1$ is the global caching gain, i.e., each transmitted message is simultaneously useful to $\frac{K M}{N}+1$ users.
 In this case, the multiplicative gap between the    resulting scheme and the MN scheme (with the memory ratio $\frac{LM}{N}$) is at most $L$.\footnote{\label{foot:gap with MN}
 Notice that if each user has an individual cache of size $LM$, then the MN scheme   would achieve the load
$
\frac{K(1 - LM/N)}{KLM/N + 1} .
$
Instead, if each user has an individual cache of size $ M$, then the MN scheme  would achieve the load in~\eqref{eq:MN}.
Therefore, the multiaccess  case model here where each user accesses to $L$ neighbouring cache-nodes,
  yields an intermediate performance in between the two extremes. }
 \item  When $L$ does not divide $K$, the authors showed that a load which is three times larger than the one in~\eqref{eq:R_HKD} can be achieved.
\end{itemize}
 %the authors obtained a scheme with transmission load $R_{\text{HKD}}=0$ by using a $(K,L)$ error correcting code for $M=N/L$, while for $M\leq N/2L$ they ignored multiaccess  and instead used the MN scheme under the assumption $L=1$. Finally, they time-shared between the two schemes for $N/2L<M<N/L$. }
The authors in \cite{RK} proposed a new scheme based on erasure-correcting code which has lower load than the scheme in \cite{JHNS} when $K<\frac{KLM}{N}+L$. In addition, when $L \geq  \frac{K}{2}$, by extending the converse bound in \cite{WTP}, a converse bound under uncoded cache placement was also proposed in \cite{RK}. Under  the constraint of $N \geq K$, $L \geq  \frac{K}{2}$, and uncoded cache placement, the   scheme   in \cite{RK} was shown to be order optimal within a factor of $2$. For some special parameters of the multiaccess  coded caching problem, some improved schemes with reduced   load compared to the scheme in \cite{RK} were proposed in~\cite{SR,MARB,SPE,Structureindexcoding}.
The authors in~\cite{OG}  extended the linear topology multiaccess  coded caching scheme~\cite{JHNS} to two-dimensional cellular networks, where the mobile users move on a two-dimensional grid and can access the nearest cache-nodes   inside the grid.
\subsection{Contribution and paper organization}
This paper focuses on the multiaccess  caching system  in \cite{JHNS}.   We show that any PDA satisfying two conditions  defined in Section \ref{sec-proof of Theorem 1} (e.g., the MN PDA and the PDAs in \cite{YCTC,SZG,CJYT,YTCC}) can be used to construct a multiaccess coded caching scheme by a novel three-step transformation. The resulting scheme has the maximum local caching gain (i.e., the cached contents stored at any $L$ neighbouring cache-nodes are different such that each user can totally retrieve $LM$ files from the connected cache-nodes) and the same coded caching gain as the shared-link coded caching scheme realized by the original PDA.
 Interestingly, the load of the resulting multiaccess coded caching scheme obtained by using the MN PDA achieves the same load as in~\eqref{eq:R_HKD} for any system parameters, i.e., without constraint that $L$ divides $K$ in \cite{JHNS}. Under the constraint of the cache placement used in this  multiaccess coded caching scheme, we show that the delivery scheme is approximately optimal when $K$ is sufficiently large. Furthermore we can reduce the   load of our resulting schemes in delivery phase by further compressing the  multicast messages.

The rest of this paper is organized as follows. The multiaccess  coded caching model and some preliminary results  are introduced in Section \ref{sec-pre}. The main results of this paper are listed in
Section \ref{sec-main-result}. Section \ref{sec-proof of Theorem 1} provides the detailed construction of the proposed caching scheme. An further improved transformation approach is described in Section \ref{sec-further}. Finally, we conclude the paper in Section \ref{sec-conclusion}, and some proofs are provided in the Appendices.
\subsection*{Notations}
%In this paper,  we use the following notations unless otherwise stated.
\begin{itemize}
\item Bold capital letter, bold lower case letter and curlicue  font  will be used to denote array, vector and set respectively.
We assume that all the sets are increasing ordered and $|\cdot|$ is used to represent the cardinality of a set or
the length of a vector;
\item For any positive integers $a$, $b$, $t$ with $a<b$ and $t\leq b $, and any nonnegative set $\mathcal{V}$,
\begin{itemize}
\item   let $[a:b]=\{a,a+1,\ldots,b\}$, especially $[1:b]$ be shorten by $[b]$, and
${[b]\choose t}=\{\mathcal{V}\ |\   \mathcal{V}\subseteq [b], |\mathcal{V}|=t\}$, i.e., ${[b]\choose t}$ is the collection of all $t$-sized subsets of $[b]$;
\item   Mod $(b,a)$ represents the modulo operation on positive integer $b$ with positive integer divisor $a$.  In this paper we let $\text{Mod}(b,a)\in \{1,\ldots,a \}$ (i.e., we let $ \text{Mod}(b,a)=a$ if $a$ divides $b$).
 $\mathcal{V}[h]$
represent the $h^{th}$ smallest element of $\mathcal{V}$, where $h\in[|\mathcal{V}|]$;
\item
Mod$(\mathcal{V},a)=\{\text{Mod}(\mathcal{V}[h],a)\ :\ h\in  [|\mathcal{V}|]\}$;
\item    $\mathcal{V}+a=\{\mathcal{V}[h]+a\ :\ h\in  [|\mathcal{V}|]\}$.
\end{itemize}
\item  the array  $[a;b]$ is written in a Matlab form, representing $\begin{bmatrix}
 a  \\
 b
\end{bmatrix} $.
\iffalse
{\color{blue}For any $K$ arrays with the same order, say $\mathbf{C}^{(1)}$, $\ldots$, $\mathbf{C}^{(K)}$, let
 $
\mathbf{C}=\left[
 \mathbf{C}^{(1)};
 \cdots;
 \mathbf{C}^{(K)}
 \right]=\left(
           \begin{array}{c}
             \mathbf{C}^{(1)} \\
             \vdots \\
             \mathbf{C}^{(K)}
           \end{array}
         \right).
$}
\fi
\end{itemize}

\section{System Model  and  Related Works}\label{sec-pre}
In this section, we first introduce the shared-link caching system in~\cite{MN}, and the MN caching scheme from the viewpoint of placement delivery array (PDA) \cite{YCTC}. %In order to more easily show how to get a multiaccess  coded caching scheme by means of an original coded caching scheme, we introduce a useful combinatorial structure called placement delivery array which can be used to characterize the original coded caching scheme. We also introduce the MN scheme from the view point of placement delivery array.
Then the multiaccess  caching model and previously known results are introduced.
\subsection{The original caching model}
\label{sub:ori model}
In the shared-link coded caching system~\cite{MN}, a server containing $N$ files with equal length in $\mathcal{W}=\{W_{1}, W_{2}, \ldots, $ $W_{N}\}$ connects through an error-free shared link to $K$ users in  $\{U_1,U_2,\ldots,U_K\}$ with $K\leq N$, and every user has a cache  which can store up to $M$ files for $0\leq M \leq N$.
An $F$-division $(K,M,N)$ coded caching scheme contains two phases.
\begin{itemize}
  \item Placement phase: Each file is divided into $F$ packets with equal size,\footnote{\label{foot:uncoded}In this paper, we only consider the uncoded cache placement.}
   and then  each user $U_k$ where $k\in [K]$ caches some packets  of each file, which is limited by its cache size $M$. Let $\mathcal{Z}_{U_k}$ denote the cache contents at user $U_k$, which is assumed to be known to the server. Notice that the placement phase is done without knowledge of later requests.
  \item Delivery phase: Each user randomly requests one file from the server. The requested file by user $U_k$ is represented by $W_{d_{U_k}}$, and the request vector by all users is denoted by $\mathbf{d}=(d_{U_1},d_{U_2},\ldots,d_{U_K})$. According to the cached contents and requests of all the users, the server transmits a broadcast message including
   $S_{{\bf d}}$   packets  to all users, such that each user's request can be satisfied.
\end{itemize}

In such system, the number of worst-case transmitted files (a.k.a. load) for all possible requests is expected to be as small as possible, which is defined as
\begin{align}
R=\max_{\mathbf{d}\in[N]^K}  \frac{ S_{\mathbf{d}}}{F}  . \label{eq:def of load}
\end{align}

%The above $(K,M,N)$ coded caching system was firstly proposed in \cite{MN}.
For the shared-link coded caching problem,
 the authors in \cite{YCTC} proposed a class of solutions based on the concept of placement delivery array (PDA) which is defined as follows.
\begin{definition}\rm(\cite{YCTC})
\label{def-PDA}
For  positive integers $K,F, Z$ and $S$, an $F\times K$ array  $\mathbf{P}=(p_{j,k})_{j\in[F] ,k\in[K]}$, composed of a specific symbol $``*"$  and $S$ positive integers from $[S]$, is called a $(K,F,Z,S)$  PDA  if it satisfies the following conditions:
\begin{enumerate}
  \item [C$1$.] The symbol $``*"$ appears $Z$ times in each column;
  \item [C$2$.] Each integer occurs at least once in the array;
  \item [C$3$.] For any two distinct entries $p_{j_1,k_1}$ and $p_{j_2,k_2}$,    $p_{j_1,k_1}=p_{j_2,k_2}=s$ is an integer  only if
  \begin{enumerate}
     \item [a.] $j_1\ne j_2$, $k_1\ne k_2$, i.e., they lie in distinct rows and distinct columns; and
     \item [b.] $p_{j_1,k_2}=p_{j_2,k_1}=*$, i.e., the corresponding $2\times 2$  subarray formed by rows $j_1,j_2$ and columns $k_1,k_2$ must be of the following form
  \begin{align*}
    \left(\begin{array}{cc}
      s & *\\
      * & s
    \end{array}\right)~\textrm{or}~
    \left(\begin{array}{cc}
      * & s\\
      s & *
    \end{array}\right).
  \end{align*}
   \end{enumerate}
\end{enumerate}
\end{definition}
 In a $(K,F,Z,S)$ PDA $\mathbf{P}$, each column represents one user's cached contents, i.e., if $p_{j,k}=*$, then user $U_k$ has cached the $j^{th}$ packet of all the files. If $p_{j,k}=s$ is an integer, it means that the $j^{th}$ packet of all the files is not stored by user $U_k$. Then the XOR of the requested packets indicated by $s$ is broadcasted by the server at time slot $s$. The property C2 of Definition \ref{def-PDA} implies that the number of signals transmitted by the server is exactly $S$. So the   load is $R=\frac{S}{F}$. Finally the property C3 of Definition \ref{def-PDA} guarantees that each user can get the requested packet, since it has cached all the other packets in the signal except its requested one.
  Hence, the following lemma was proved by  Yan et al. in \cite{YCTC}.
\begin{lemma}\rm(\cite{YCTC})
\label{le-Fundamental}Using Algorithm \ref{alg:PDA}, an $F$-division caching scheme for the $(K,M,N)$ caching system can be realized by a $(K,F,Z,S)$ PDA  with $\frac{M}{N}=\frac{Z}{F}$. Each user can decode his requested file correctly for any request ${\bf d}$ at the rate $R=\frac{S}{F}$.
\hfill $\square$
\end{lemma}
%Given a $(K,F,Z,S)$ PDA $\mathbf{P}=(p_{j,k})$, an $F$-division $(K,M,N)$ caching scheme with $\frac{M}{N}=\frac{Z}{F}$ and $R=\frac{S}{F}$  can be obtained by.
\begin{algorithm}[htb]
\caption{Shared-link coded caching scheme  based on the PDA in \cite{YCTC}}\label{alg:PDA}
\begin{algorithmic}[1]
\Procedure {Placement}{$\mathbf{P}$, $\mathcal{W}$}
\State Split each file $W_n\in\mathcal{W}$ into $F$ packets, i.e., $W_{n}=(W_{n,j}\ :\ j\in [F])$.
\For{$k\in [K]$}
\State $\mathcal{Z}_{U_{k}}\leftarrow\{W_{n,j}\ |\ p_{j,k}=*, \forall~n\in[N]\}$
\EndFor
\EndProcedure
\Procedure{Delivery}{$\mathbf{P}, \mathcal{W},{\bf d}$}
\For{$s=1,2,\cdots,S$}
\State  Server sends $\bigoplus_{p_{j,k}=s,j\in[F],k\in[K]}W_{d_{U_k},j}$.
\EndFor
\EndProcedure
\end{algorithmic}
\end{algorithm}

Let us briefly introduce the MN coded caching scheme in~\cite{MN} from the viewpoint of   MN PDA, where the resulting PDA is referred to as MN PDA. For any integer $t\in[K]$,  we let $F={K\choose t}$. We arrange all the subsets with size $t+1$ of $[K]$ in the lexicographic order and define $\phi(\mathcal{S})$ to be its order for any subset $\mathcal{S}$ of size $t+1$. Clearly, $\phi$ is a bijection from ${[K]\choose t+1}$ to $[{K\choose t+1}]$.
Then, an MN PDA is defined as a ${K\choose t}\times K$ array $\mathbf{P}=(p_{\mathcal{T},k})_{\mathcal{T}\in {[K]\choose t},k\in [K]}$
by
\begin{align}\label{Eqn_Def_AN}
p_{\mathcal{T},k}=\left\{\begin{array}{cc}
\phi(\mathcal{T}\cup\{k\}), & \mbox{if}~k\notin\mathcal{T}\\
*, & \mbox{otherwise}
\end{array}
\right.
\end{align}
where the rows are labelled by all the subsets $\mathcal{T}\in {[K]\choose t}$. Thus, by the definition of PDA, the achieved load of the MN PDA is as follows.
\begin{lemma}\rm(MN PDA\cite{MN})
\label{le-MN}
For any positive integers $K$ and $t$ with $t<K$, there exists a $\left(K,{K\choose t}, {K-1\choose t-1}, {K\choose t+1}\right)$ PDA, which gives a ${K\choose t}$-division $(K,M,N)$ coded caching scheme for the shared-link caching system with the memory ratio $\frac{M}{N}=\frac{t}{K}$ and   load $R=\frac{K-t}{t+1}$.
\hfill $\square$
\end{lemma}

\begin{example}\label{MN-pda}
\rm
We then illustrate the MN PDA by the following example where $N=K=4$ and  $M=2$.
By the construction of the MN PDA, we have  the following   $(4,6,3,4)$ PDA.
\begin{align}
\label{eq-PDA-6-4}
\mathbf{P}=\left(\begin{array}{cccc}
*	&	*	&	1	&	2	\\
*	&	1	&	*	&	3	\\
*	&	2	&	3	&	*	\\
1	&	*	&	*	&	4	\\
2	&	*	&	4	&	*	\\
3	&	4	&	*	&	*	
\end{array}\right).
\end{align}
Using Algorithm \ref{alg:PDA}, the detailed caching scheme is as follows.
 %one can obtain a $6$-division $(4,3,6)$ coded caching scheme for the original MN caching model in the following way.
\begin{itemize}
   \item \textbf{Placement Phase}: From Line 2 we have $W_n=(W_{n,1},W_{n,2},W_{n,3},W_{n,4},W_{n,5},W_{n,6})$ where  $n\in [4]$. Then by Lines 3-5, the users' caches are
       \begin{align*}
       \mathcal{Z}_{U_1}=\left\{W_{n,1},W_{n,2},W_{n,3}\ :\ n\in[4]\right\};\ \ \ \ \ \ \
       \mathcal{Z}_{U_2}=\left\{W_{n,1},W_{n,4},W_{n,5}\ :\ n\in[4]\right\}; \\
       \mathcal{Z}_{U_3}=\left\{W_{n,2},W_{n,4},W_{n,6}\ :\ n\in[4]\right\};\ \ \ \ \ \ \
       \mathcal{Z}_{U4}=\left\{W_{n,3},W_{n,5},W_{n,6}\ :\ n\in[4]\right\}.
       \end{align*}
   \item \textbf{Delivery Phase}: Assume that the request vector is $\mathbf{d}=(1,2,3,4)$. By the transmitting process by Lines 8-10,  the server transmits the following multicast messages with total load $R=\frac{4}{6}=\frac{2}{3}$.
   \begin{table}[!htp]
  \normalsize{
  \begin{tabular}{|c|c|}
\hline
    % after \\: \hline or \cline{col1-col2} \cline{col3-col4} ...
   Time Slot& Transmitted Signnal  \\
\hline
   $1$&$W_{1,4}\oplus W_{2,2}\oplus W_{3,1}$\\ \hline
   $2$&$W_{1,5}\oplus W_{2,3}\oplus W_{4,1}$\\ \hline
  $3$& $W_{1,6}\oplus W_{3,3}\oplus W_{4,2}$\\ \hline
   $4$& $W_{2,6}\oplus W_{3,5}\oplus W_{4,4}$\\ \hline
  \end{tabular}}\centering
  \caption{Delivery steps in Example \ref{MN-pda} }\label{table1}
\end{table}
\end{itemize}
\hfill $\square$
\end{example}

\subsection{Multiaccess  coded caching model}
\label{sub:multiaccess model}
We then introduce the $(K,L,M,N)$    multiaccess  coded caching problem in \cite{JHNS} (as illustrated in Fig.~\ref{multi-access-system}), containing  a server with  a set of $N$ equal-length files  (denoted by $\mathcal{W}=\{W_{1},  \ldots,   W_{N}\}$), $K$ cache-nodes (denoted by $ C_1, \ldots,C_{K} $), and $K\leq N$ users (denoted by $ U_{1},  \ldots, U_{K} $). Each cache-node has a memory size of $M$ files where $0\leq M \leq  \frac{N}{L}$. Each user is connected to $L$ neighbouring cache-nodes in a cyclic wrap-around fashion, where user $U_k$ can access the cache-nodes $C_{k},C_{\text{Mod}(k+1,K)}, \ldots, C_{\text{Mod}(k+L-1,K)}$ for each $k\in [K]$.
Each user  is also connected via an error-free shared link to the server. In this paper, we assume that the communication bottleneck is on the shared link from the server to the users; thus we assume that each user can  retrieve   the cached contents from its connected cache-nodes without any cost.
\iffalse
\begin{figure}
\centering
\includegraphics[width=3in]{multi-access-system}
\vskip 0.2cm
\caption{Multiaccess  network}\label{multi-access-system}
\end{figure}
\fi

An $F$-division $(K,L,M,N)$ multiaccess  coded caching scheme runs in two phases,
\begin{itemize}
\item Placement phase: each file is divided into $F$ packets of equal size, and then each cache-node $C_k$ where $k\in[K]$, directly caches some packets of each file, which is limited by its cache size $M$. Let $\mathcal{Z}_{C_k}$ denote the cache contents at cache-node $C_k$. The placement phase is also done without knowledge of later requests. Each user $U_{k}$ where $k \in [K]$ can retrieve   the contents cached at the $L$ neighbouring cache-nodes in a cyclic wrap-around fashion. Let $\mathcal{Z}_{U_k}$ denote the retrievable cache contents  by user $U_k$. %$\mathcal{Z}_{C_k}$ and $\mathcal{Z}_{U_k}$ are assumed to be known by the server.
\item Delivery phase: each user randomly requests one file.  According to the request vector $ \mathbf{d}=(d_{U_1},d_{U_2},\ldots,d_{U_K})$ and the cached contents in the cache-nodes, the server transmits $S_{{\bf d}}$ multicast messages to all users, such that each user's request can be satisfied.
\end{itemize}

We aim  to design a multiaccess  coded caching scheme with minimum   worst-case  load as defined in~\eqref{eq:def of load}.

Note that any uncoded cache placement (as described above), we can divide each file $W_{n}$, $n\in [N]$ into  subfiles, $W_{n}=\{W_{n,\mathcal{T}}:\mathcal{T} \subseteq [K] \}$. Subfile $W_{n,\mathcal{T}}$ represents the set of packets of $W_n$ cached by cache-nodes $C_k$  where $k\in \mathcal{T}$.

The first multiaccess  coded caching scheme was proposed in \cite{JHNS},   where the following result was proved.
\begin{lemma}[\cite{JHNS}]\rm
\label{le-JHNS}
For the   $(K,L,M,N)$    multiaccess  coded caching problem,
\begin{itemize}
\item when $L|K$, the lower convex envelope of the memory-load tradeoff points
\begin{align}
(M,R_{\text{HKD}})= \left(\frac{N t}{K}, \frac{K-t L}{t+1} \right), \ \forall t\in \left\{0,1, \ldots,  \left\lfloor
\frac{K}{L} \right\rfloor  \right\} \label{eq:division load}
\end{align}
and $(M,R_{\text{HKD}})=\left(\frac{N}{L},0 \right)$, is achievable.
\item When $L\nmid K$, the lower convex envelope of the memory-load tradeoff points
\begin{align}
(M,R_{\text{HKD}})= \left(\frac{N t}{K}, \frac{K-  t}{t+1} \right), \ \forall t\in \left\{0,1, \ldots,  \left\lfloor
\frac{K}{2L} \right\rfloor  \right\}
\end{align}
and $(M,R_{\text{HKD}})=\left(\frac{N}{L},0 \right)$, is achievable.
\end{itemize}
\hfill $\square$
 \end{lemma}
When $L|K$, the authors in \cite{JHNS} showed that the multiplicative gap between the $R_{\text{HKD}}$ and the MN scheme with the memory ratio $\frac{LM}{N}$ is at most $L$. Based on Minimum Distance Separable (MDS) codes, the authors in \cite{RK} proposed the following improved scheme.
\begin{lemma}[\cite{RK}]
\label{lamm-2}\rm
For the   $(K,L,M,N)$    multiaccess  coded caching problem, the lower convex envelope of the memory-load tradeoff points
\begin{align}
(M,R_{\text{RK}})  =  \left(\frac{N t}{K}, \frac{(K-tL)^2}{K}  \right), \ \forall t\in \left\{0,1, \ldots,  \left\lfloor
\frac{K}{L} \right\rfloor  \right\}     \label{eq:lem2 load}
\end{align}
is achievable.
\hfill $\square$
 \end{lemma}
%When the subpacketization is at most $K(K-1)$, the authors in \cite{SR} proposed the following scheme with lower transmission load compared with the scheme in \cite{RK}, expect for the optimal schemes proposed in \cite{RK}.
 The load in~\eqref{eq:lem2 load} was shown in \cite{RK} to be optimal under uncoded cache placement when $L=K-1$; $L=K-2$;  $L=K-3$ for $K$ is even. The authors in \cite{SR} proposed a scheme with the subpacketization $F=K(K-1)$  achieving the  following load.
\begin{lemma}[\cite{SR}]
\label{lamm-3}\rm
For the   $(K,L,M,N)$    multiaccess  coded caching problem, the lower convex envelope of the memory-load tradeoff points   $\left(\frac{N t}{K},  R_{\text{SR}}\right)$ for all
  $t   \in \left\{0,1, \ldots,  \left\lfloor
\frac{K}{L} \right\rfloor  \right\} $ and $\left(\frac{N}{L},0 \right)$, is achievable, where
\begin{itemize}
\item $R_{\text{SR}}=\frac{1}{K}$ if $K-1=tL$;
\item $R_{\text{SR}}=\sum_{h=\frac{K-tL+2}{2}}^{K-tL}\frac{2}{1+\lceil \frac{tL}{h}\rceil}$ if $K-tL$ is even;
\item $R_{\text{SR}}=\frac{1}{\lceil\frac{2tL}{K-tL+1}\rceil+1}+\sum_{h=\frac{K-tL+3}{2}}^{K-tL}\frac{2}{1+\lceil\frac{ tL}{h} \rceil}$ if $K-tL>1$ is odd.
\end{itemize}
\hfill $\square$
\end{lemma}
When $M=\frac{2 N}{K}$ (i.e., $t=\frac{K M}{N}=2$), the authors in~\cite{SPE} proposed a caching scheme with load $\frac{K-t L}{g}$ where $g>t+1 =3$,  which is strictly lower than the load in~\eqref{eq:division load}.
When $M=\frac{N}{K}$, the authors in \cite{MARB} proposed a caching scheme with  linear subpacketization, but achieving a higher load than the above schemes when $L < \frac{K}{2}$.
\iffalse
\begin{lemma}[\cite{MARB}]
\label{lamm-4}\rm
For the   $(K,L,M,N)$    multiaccess  coded caching problem where   $M=\frac{N}{K}$, the following load is achievable,
\begin{align*}
R_{MS}=\left\lceil \frac{K(K-L)}{2+\lfloor \frac{L}{K-L+1} \rfloor+ \lfloor \frac{L-1}{K-L+1} \rfloor} \right\rceil \frac{1}{K}.
\end{align*}
\hfill $\square$
\end{lemma}
\fi
\section{Main Results}
\label{sec-main-result}
\subsection{Main results and performance analysis}
%In this section, {\color{blue}by our non-trivial transformation, which will introduce detailed in Section \ref{sec-proof of Theorem 1},} from the MN PDA, we propose a novel multiaccess coded caching scheme, which achieves the same load as the scheme in \cite{JHNS} but for any parameters $K$ and $L$.
 In this section, a novel multiaccess coded caching scheme can be obtained by a non-trivial transformation from an appropriate PDA satisfying some properties. %As applications, based on the MN PDA, we obtain one class of multiaccess coded caching scheme. It is worth noting that the scheme achieves the same load as the scheme in \cite{JHNS} but for any parameters $K$ and $L$.}
By applying the novel transformation to the MN PDA, we obtain the following theorem, whose detailed proof could be found in Section \ref{sec-proof of Theorem 1}.
%The performance of the new proposed scheme is given the following,
\begin{theorem}\rm
\label{th-1}
For the $(K,L,M,N)$    multiaccess  coded caching problem,
the lower convex envelope of the following memory-load tradeoff corner points are achievable,
\begin{align}
(M, R_{1})=  \left(\frac{N t}{K}, \frac{K-tL}{t+1} \right), \ \forall t  \in  \left[0 :  \left\lfloor
\frac{K}{L} \right\rfloor  \right] , \label{eq:load R1}
\end{align}
and $(M,R_1)=\left( \frac{N}{L}, 0\right)$.  The subpacketization is $F=K{K-t(L-1)\choose t}$ when $M=\frac{N t}{K}$, for $t\in \left[0 :   \left\lfloor
\frac{K}{L} \right\rfloor  \right]$.
\hfill $\square$
\end{theorem}
Note that   the scheme in Theorem~\ref{th-1} achieves the same load as the scheme in \cite{JHNS} but for any parameters $K$ and $L$.
As the existing schemes in Lemmas~\ref{le-JHNS}-\ref{lamm-3},
the non-trivial corner points of the proposed scheme are at the memory sizes
$M=\frac{N t}{K}$ where $t \in \left\{ 1,\ldots,  \left\lfloor
\frac{K}{L} \right\rfloor  \right\}$.

We then compare the achieved load  by the proposed scheme in Theorem~\ref{th-1} with the existing schemes in Lemmas~\ref{le-JHNS}-\ref{lamm-3}.
\begin{itemize}
\item Comparison to Lemma~\ref{le-JHNS}.
\begin{enumerate}
\item When $L$ divides $K$, the proposed scheme in Theorem~\ref{th-1} achieves the same load as Lemma~\ref{le-JHNS}, i.e., $R_1=R_{\text{HKD}}$.
\item When $L$ does not divide $K$, we have $\frac{R_{\text{HKD}}}{R_1}=\frac{K-t}{K-Lt}$ for $M=\frac{N t}{K}$   where $t \in \left\{0,1, \ldots,  \left\lfloor
\frac{K}{2L} \right\rfloor  \right\} $. In addition, for $\frac{N}{K} \left\lfloor
\frac{K}{2L} \right\rfloor  \leq M <\frac{N}{L}$, we have
    \begin{align}
    \label{eq:comparison to HKD}
    \frac{R_{\text{HKD}}}{R_1}> \frac{K-\lfloor\frac{K}{2L}\rfloor}{K-L\lfloor\frac{K}{2L}\rfloor}
    \end{align}
 which will be proved in Appendix~\ref{sec:proof of HKD}.
\end{enumerate}
\item  Comparison to Lemma~\ref{lamm-2}. For the non-trivial corner points with $t=\frac{K M}{N} \in \left\{ 1,\ldots,  \left\lfloor
\frac{K}{L} \right\rfloor  \right\}$, it will be proved in Appendix~\ref{sec:proof of RK} that
\begin{align}
 R_{1} <  R_{\text{RK}} , \ \text{if }  K >  (t+1)L. \label{eq:comparison to RK}
\end{align}
\item  Comparison to Lemma~\ref{lamm-3}. Due to the high complexity of the closed-form of the  load in Lemma~\ref{lamm-3}, we cannot provide the exact comparison between the proposed scheme and the scheme in Lemma~\ref{lamm-3}. Instead, in Appendix~\ref{sec:proof of SR} we will show that for the non-trivial corner points with $t=\frac{K M}{N} \in \left\{ 1,\ldots,  \left\lfloor
\frac{K}{L} \right\rfloor  \right\}$, we have
\begin{align}
 R_{1}< R_{\text{SR}} , \ \text{if }  K \gg L.\label{eq:comparison to SR}
\end{align}
\end{itemize}
Note that the designed placement of the proposed scheme in Theorem~\ref{th-1} is identical to the placements in~\cite{JHNS,RK} in the sense of subfile division. More precisely, for $M=\frac{N t}{K}$ where $t\in [\left\lfloor
\frac{K}{L} \right\rfloor ]$,  each file $W_n$, $n\in [N]$ is divided into $\binom{K-tL+t-1}{t-1}\frac{K}{t}$ non-overlapping and equal-length subfiles, where
\begin{align}
W_n=\{W_{n,\mathcal{T}}: \mathcal{T} \subseteq [K], |\mathcal{T}|=t, \text{Mod}(j_1-j_2,K)\geq L, \ j_1,j_2 \in \mathcal{T} \text{ and } j_1 \neq j_2  \};\label{eq:our subfile division}
\end{align}in words, each subfile is cached by $t$ cache-nodes and any two cache-nodes within distance $L-1$ do not cache any common packets.
The optimal load under the   cache placement with subfile division in~\eqref{eq:our subfile division} is denoted by $R^{\star}_1$. Under this placement, we propose a converse bound on the load,   and prove that   when $K \gg tL$ the proposed delivery scheme in  Theorem~\ref{th-1} is approximately optimal, where the detailed proof could be found in Appendix~\ref{sec:converse proof}.
\begin{theorem}\rm
\label{thm:converse}
 For the $(K,L,M,N)$ multiaccess coded caching problem   with  $M=\frac{N t}{K}$ where $t\in [\left\lfloor
\frac{K}{L} \right\rfloor ]$, it holds that
\begin{subequations}
 \begin{align}
& R^{\star}_1 \geq \frac{\binom{K-L-X}{t-1}+(K-L-1)\binom{K-L-X+1}{t} - \binom{K-L-X}{t+1} }{K \binom{K-X}{t-1}}, \label{eq:converse bound}\\
&{\text where } \ X:= t L-t+1. \label{eq:def of X}
 \end{align}
 When $K \gg t L$, we have
 \begin{align}
 \frac{ R_{1} }{R^{\star}_1 } \approx 1. \label{eq:approximate opt}
 \end{align}
 \end{subequations}
 \hfill $\square$
\end{theorem}

It will be explained in Section~\ref{sub:illustrated example} and Section~\ref{sec-proof of Theorem 1}, the main novelty of the proposed scheme in Theorem~\ref{th-1} is a smart transformation of the MN PDA for the multiaccess caching problem.  In Section~\ref{subsect-delivery-user} we will point out any PDA
satisfying the conditions C$4$ and C$5$ in Proposition \ref{pro-fundament} can be used to generate the scheme for the multiaccess caching problem, e.g., the PDAs in~\cite{YCTC,SZG,CJYT,YTCC}. In one word, based on the original PDA, we design a multiacess coded caching scheme with maximal local caching gain (i.e., any  $L$ neighbouring cache-nodes do not cache any same packet such that each user can retrieve $LMF$ packets from its connected cache-nodes) and the same coded caching gain as the original PDA.
{\bf As a result, our proposed transformation approach has advantages both in terms of load and in terms of subpacketization.
}
For instance,
%We also take the PDA in \cite{YCTC} as another example to show our claim. By the same transformation approach, some other existing PDAs for the original MN problem (e.g., the PDAs in~\cite{SZG,CJYT}) can be also extended to the multiaccess  caching problem, as long as C$4$ and C$5$ are satisfied.}
to apply the transformation approach into the PDA in~\cite{YCTC}, we obtain the following multiaccess coded caching scheme  with a lower subpacketization, whose proof is in Appendix \ref{sec:proof of Partition theorem}.
\begin{theorem}\rm
\label{th-3}
For any positive integers $m$, $q\geq 2$, there exists a $(K=m(q+L-1),L,M,N)$ multiaccess coded caching scheme with the memory ratio $\frac{M}{N}=\frac{1}{q+L-1}$, subpacketization $F=mq^{m-1}(q+L-1)$ and transmission load $R_2=q-1$.\hfill $\square$
\end{theorem}

We conclude this subsection with some numerical evaluations to compare the proposed schemes in Theorems~\ref{th-1} and \ref{th-3}, the derived converse in Theorem \ref{thm:converse} and the existing schemes in Lemmas~\ref{le-JHNS}-\ref{lamm-3}.
In Fig.~\ref{Fig-Compare} and Fig.~\ref{Fig-Compare2}, we consider the multiaccess  caching problem with $K=N=20$, $L=3,4$.
When $L=3$,
it can be seen that the proposed scheme performs the best when  $0\leq \frac{M}{N} \leq  \frac{K/L-1}{K}=\frac{20/3-1}{20}=\frac{1}{3}-\frac{1}{20}\approx 0.28$ by the condition $K>(t+1)L$ in \eqref{eq:comparison to RK}.
When $L=4$,
it can be seen that the proposed scheme performs the best when $0\leq \frac{M}{N}\leq 0.217$. %{\color{blue}In addition, in Fig.~\ref{fig:numerical}, for most of memory ratios, the transmission load of the scheme in Theorem \ref{th-3} is smaller than those of the schemes in Lemmas~\ref{le-JHNS}-\ref{lamm-3}. This can be directly obtained by the fact that the transmission load of the scheme in Theorem \ref{th-3} approximates that of the scheme in Theorem \ref{th-1}. The detailed relationship between the PDA in \cite{YCTC} and the MN PDA can be found in \cite[Section VI-A Comparison With Ali-Niesen Scheme]{YCTC}.}

\begin{figure}[http!]
\centering
\includegraphics[width=4in]{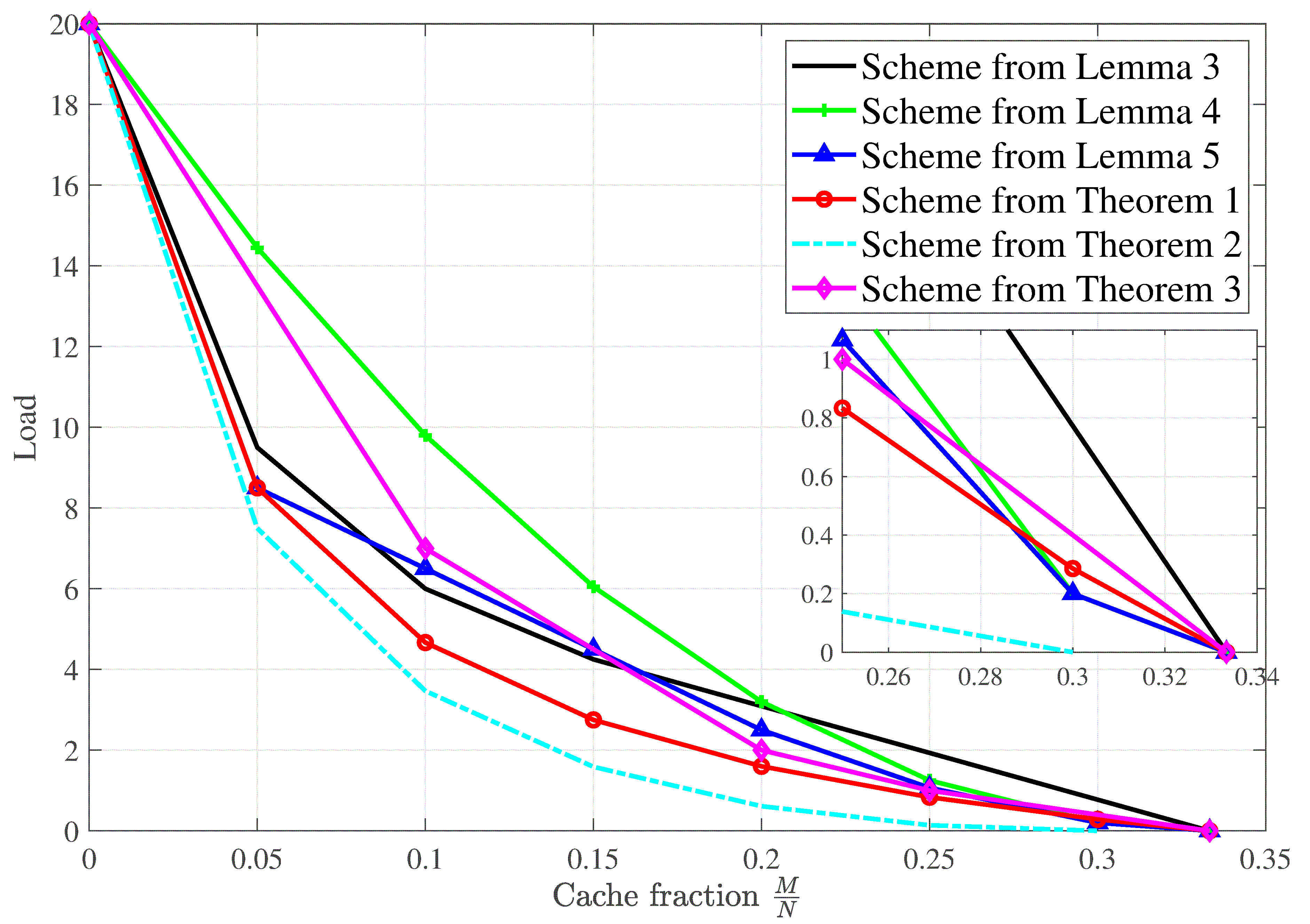}
\vskip 0.2cm
\caption{Comparison of the caching schemes  for the    multiaccess  coded caching problem with $N=K=20$ and $L=3$.}\label{Fig-Compare}
\end{figure}

\begin{figure}[http!]
\centering
\includegraphics[width=4in]{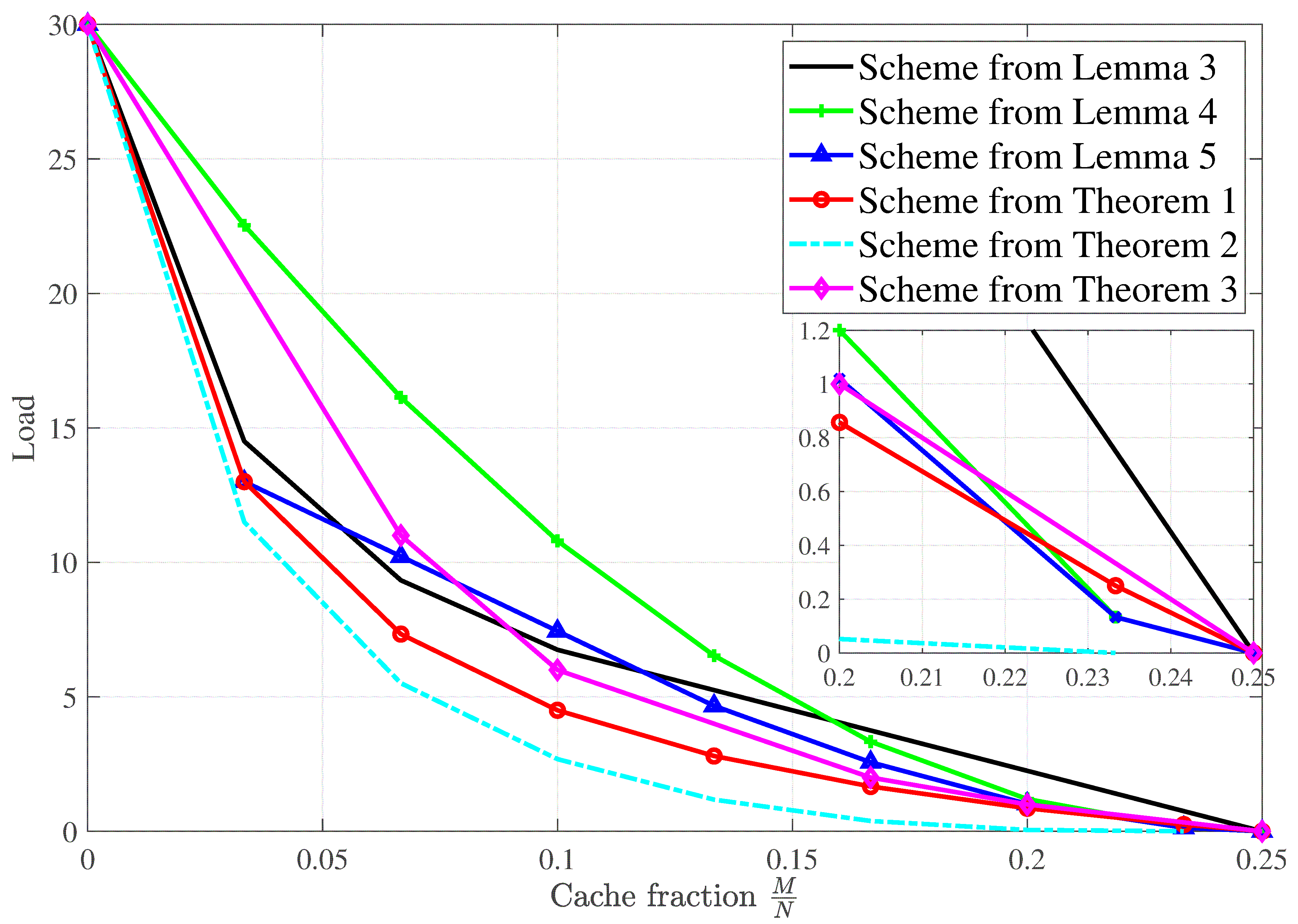}
\vskip 0.2cm
\caption{Comparison of the caching schemes  for the    multiaccess  coded caching problem with $N=K=30$ and $L=4$.}\label{Fig-Compare2}
\end{figure}
\subsection{Sketch of the proposed scheme in Theorem~\ref{th-1}}
\label{sub:illustrated example}
Let us consider the $(K,L,M,N)=(8,3,2,8)$ multiaccess  caching problem where user $U_k$ can access the cache-nodes $C_{k}$, $C_{\text{Mod}(k+1,8)}$ and  $C_{\text{Mod}(k+2,8)}$ for each $k\in [8]$. We aim to propose a multiaccess coded caching scheme based on the $(K',F' ,Z ,S )=(4,6,3,4)$ MN PDA, i.e.,  $\mathbf{P}$ in \eqref{eq-PDA-6-4}. It will be clear in Section~\ref{sec-proof of Theorem 1} that we choose the MN PDA where $K'=K-\frac{KM}{N}(L-1)$ and $K'\frac{Z}{F'}=K\frac{M}{N}$. Generally speaking,  the main  ingredients of the proposed scheme contain two key points:
\begin{itemize}
\item According to the multiaccess  caching model, in order to fully use the cache-nodes, we design a cache placement such that any two cache-nodes connected to some common user(s) should not cache the same packets. Since each user can access $L$ cache-nodes, the local caching gain of the proposed scheme is
\begin{align*}
g_{\text{local}}= 1-\frac{L M}{N}= \frac{1}{4}.
\end{align*}
\item The  structure of the proposed  placement and delivery phases is based on the MN PDA, such that the coded caching gain is the same as the    $(4,6,3,4)$ MN PDA, which is equal to
\begin{align*}
 g_{\text{coded}}= \frac{K' Z }{F' }+1= \frac{K M}{N}+1=3.
\end{align*}
\end{itemize}
Thus the   load of the proposed scheme is
 \begin{align*}
 K \frac{g_{\text{local}}}{ g_{\text{coded}}} = \frac{2}{3},
\end{align*}
which coincides with~\eqref{eq:load R1}. Then we introduce the more details on the construction. We divide each file into $K=8$ parts with equal length, $W_n= \left(W^{(g)}_n: g\in [8] \right)$, and divide the caching procedure into $8$ rounds, where in the $g^{\text{th}}$ round we only consider the $g^{\text{th}}$ part of each file. We should point out that in order to ensure all the cache-nodes caching the same amount of contents, we have to divide the caching procedure into $K$ rounds.
Furthermore, it will be clarified later that the caching procedures for different rounds are totally symmetric, and thus we only need to design  the caching scheme for the first round (i.e., for the first subfiles of the $N$ files).

Let us consider the first round to illustrate the main idea.
We further divide the first part of each file (i.e., $W^{(1)}_n$ for each $n\in [8]$) into $F'=6$  non-overlapping and equal-length packets.
  We then construct     three arrays $\mathbf{C}^{(1)}$, $\mathbf{U}^{(1)}$, and $\mathbf{Q}^{(1)}$ via $\mathbf{P}$, defined as follows.
\begin{definition}
\label{defn:three arrays}
For each $g\in [K]$,
\begin{itemize}
\item An $F' \times K$ node-placement array $\mathbf{C}^{(g)}$ consists of star and null, where $F'$ and $K$ represent the subpacketization  of each part and the number of cache-nodes,  respectively. The entry  located at the position $(j,k)$ in  $\mathbf{C}^{(g)}$ is star if and only if the $k^{\text{th}}$ cache-node   caches  the $j^{\text{th}}$ packet of each $W^{(g)}_n$ where $n\in [N]$. %otherwise this cache-node does not cache any of these packets;
\item An $F' \times K$ user-retrieve array $\mathbf{U}^{(g)}$ consists of star and null, where $F'$ and $K$ represent the subpacketization of each part and the number of users, respectively. The entry at the position $(j,k)$ in  $\mathbf{U}^{(g)}$
is star if and only if the $k^{\text{th}}$ user can retrieve   the $j^{\text{th}}$ packet  of each $W^{(g)}_n$ where $n\in [N]$.
\item An $F' \times K$ user-delivery array $\mathbf{Q}^{(g)}$ consists of $\{*\}\cup[S]$, where $F'$, $K$ and the stars in $\mathbf{Q}^{(g)}$ have the same meaning as $F'$, $K$ of $\mathbf{U}^{(g)}$ and the stars in $\mathbf{U}^{(g)}$, respectively. Each integer represents a multicast message, and $S$ represents the total number of multicast messages transmitted in the $g^{\text{th}}$ round of the delivery phase.
\end{itemize}
\hfill $\square$
\end{definition}
The constructing flow diagram from  $\mathbf{P}$ to $\mathbf{Q}^{(1)}$ (listed in Fig.~\ref{Fig-construct}) contains the following three steps:
 \begin{figure}
\centering
\includegraphics[width=7in]{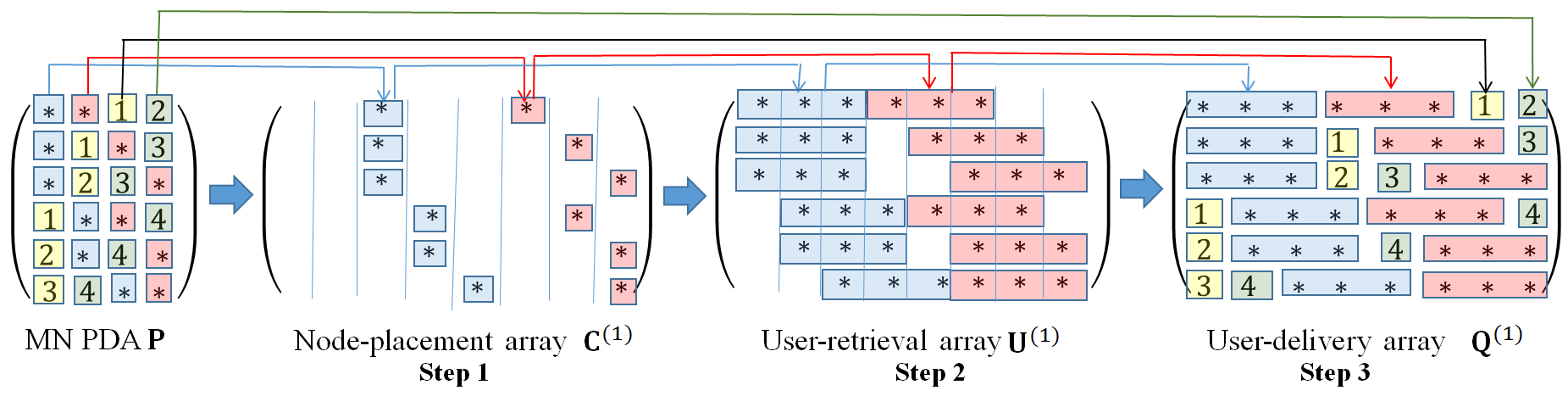}
\vskip 0.2cm
\caption{The flow diagram of constructing $\mathbf{C}^{(1)}$, $\mathbf{U}^{(1)}$  and $\mathbf{Q}^{(1)}$ via $(4,6,3,4)$ MN PDA}\label{Fig-construct}
\end{figure}
  \begin{itemize}
  \item {\bf Step 1}. In the first step, we construct the node-placement array $\mathbf{C}^{(1)}$ from  $\mathbf{P}$. More precisely,
  let us focus on each row $j\in[6]$  of $\mathbf{P}$. It can be seen that row $j$ contains two stars, which are assumed to be located at columns $k_1$ and $k_2$, respectively.
  We let row $j$  of $\mathbf{C}^{(1)}$ also contain two stars, where the first star is located at column $k_1+(L-1)=k_1+2$ and the second star is located at column $k_2+2(L-1)=k_2+4$.
    For example, as the blue line showed in Step 1, the   $*$ at position $(1,1)$ of $\mathbf{P}$ is put at position $(1,1+2)=(1,3)$ of $\mathbf{C}^{(1)}$, and
 the   $*$ at position $(1,2)$ of $\mathbf{P}$ is put at position $(1,2+4)=(1,6)$ of $\mathbf{C}^{(1)}$. {\bf By this construction, any  $L$  neighbouring cache-nodes do not cache any common packets.}
  \item {\bf Step 2}. In the second step, we construct the user-retrieve array $\mathbf{U}^{(1)}$ from $\mathbf{C}^{(1)}$. More precisely, since each user can access $L$ neighbouring cache-nodes in a cyclic wrap-around fashion, we put the stars in $\mathbf{U}^{(1)}$ according to position of the stars in $\mathbf{C}^{(1)}$.  In other words, if the entry at the position $(j,k)$  of $\mathbf{C}^{(1)}$ is star, then the entry at the position $(j,k_1)$ of $\mathbf{U}^{(1)}$ is set to be star where $0\leq k-k_1\leq L-1=2$.
    \item {\bf Step 3}. In the third step, we construct the user-delivery array $\mathbf{Q}^{(1)}$ from $\mathbf{U}^{(1)}$. First, we let $\mathbf{Q}^{(1)}=\mathbf{U}^{(1)}$. Recall that the null entries in  the $k^{\text{th}}$ column represent the required  packets in $W^{(1)}_{d_{U_k}}$ which can not be retrieved by user $k$ from its connected cache-nodes.
 For example, the entry at the position $(1,7)$ is null, because user $U_7$ requires the first packet of $W^{(1)}_{d_{U_7}}$ which can not be retrieved   from  its connected cache-nodes $C_7$, $C_8$ and $C_1$. %So we only need to consider the null entries in $\mathbf{Q}^{(1)}$.
 We then fill the null entries in $\mathbf{Q}^{(1)}$ by the integers in $\mathbf{P}$ following the delivery strategy of the MN PDA.
 Let us focus on the   row  $j\in [6]$ of $\mathbf{P}$ and  $\mathbf{Q}^{(1)}$.  Notice that   row $j$ of $\mathbf{P}$  contains two integers, and that row  $j$  of  $\mathbf{Q}^{(1)}$ also contains two nulls. Thus we set the first null to be the first integer and the second null to be the second integer.
  \end{itemize}
  {\bf After determining $\mathbf{Q}^{(1)}$, it can be seen  that $\mathbf{Q}^{(1)}$  satisfies the condition C$3$ of Definition \ref{def-PDA}.}
  Hence, we use the delivery strategy in Line 9 of Algorithm \ref{alg:PDA}. For example,  assume that the request vector is ${\bf d}=(1,2,\ldots,8)$. In the first transmission which corresponds to the integer `$1$' in $\mathbf{Q}^{(1)}$, the server sends the XOR of  the fourth packet of $W^{(1)}_{1}$  (denoted by  $W^{(1)}_{1,4}$),  the second packet of $W^{(1)}_{4}$   (denoted by $W^{(1)}_{4,2}$),  and the first packet of $W^{(1)}_{7}$ (denoted by $W^{(1)}_{7,1}$), i.e., $W^{(1)}_{1,4}\bigoplus W^{(1)}_{4,2}\bigoplus W^{(1)}_{7,1}$.

Finally for each $g\in[8]$,  the arrays $\mathbf{C}^{(g)}$ and  $\mathbf{Q}^{(g)}$ can be obtained by  cyclically  right-shifting $\mathbf{C}^{(1)}$ and $\mathbf{Q}^{(1)}$ by $g-1$ positions, respectively. After obtaining $\mathbf{C}^{(g)}$ and $\mathbf{Q}^{(g)}$, the placement and delivery phases in the $g^{\text{th}}$ round can be   done as above.

Let
$
\mathbf{C}=\left[
              \mathbf{C}^{(1)};
              \mathbf{C}^{(2)};
              \cdots;
              \mathbf{C}^{(8)}
            \right]
$
be the total placement array of the cache-nodes.
By the construction, each column of $\mathbf{C}$ is a concatenation of all the columns of $\mathbf{C}^{(1)}$; thus
each column of $\mathbf{C}$ has exactly $12$ stars. In addition, the array $\mathbf{C}$ contains $6\times 8=48$ rows.
Hence, the needed memory size is
$
M=  \frac{12N}{48}=2,
$
satisfying the memory size constraint.
The  total subpacketization of the proposed scheme is $48$, which is equal to the number of rows in $\mathbf{C}$.

\iffalse
It can be seen  from the above example that, if we want to extend some existing PDA (represented by the array $\mathbf{P}$) for the original MN caching problem to the multiaccess model by the proposed transformation approach, the PDA $\mathbf{P}$ should satisfy that (i) the number of stars in each row of P is the same; (ii) the resulting array $Q^{(1)}$ from $\mathbf{P}$ satisfies Condition C3 in Definition~\ref{def-PDA},  which are the conditions C$4$ and C$5$ introduced in Section \ref{sec-proof of Theorem 1}.
\fi
 Besides the MN PDA,   we can also apply the proposed transformation approach to
 the $(K',F' ,Z ,S )=(4,2,1,2)$ PDA in  \cite[Theorem 4 ($m=1$ and $q=2$)]{YCTC} or \cite[Theorem 3 ($m=1$ and $q=2$, $z=1$)]{CJYT}  to  obtain a multiaccess coded caching scheme with the subpacketization $F' \times K=16$ and load $\frac{S}{F'}=1$;
to     the $(K',F' ,Z ,S )=(4,4,2,4)$ PDA in  \cite[Theorem 18 ($m=2$ and $q=2$, $t=1$)]{SZG},to  obtain a multiaccess coded caching scheme  with  the  subpacketization $F' \times K=32$ and load $\frac{S}{F'}=1$.

\section{Proofs  of Theorem \ref{th-1} and Theorem~\ref{th-3}}
\label{sec-proof of Theorem 1}
%In this section, for the ease of understanding,   we introduce our construction of multiaccess  coded caching scheme based on the MN PDA. It will be clear  that besides of the MN PDA, other PDAs for the original MN caching problem can also be used to construct multiaccess  coded caching scheme as long as  they satisfy the conditions of Proposition \ref{th-main} given in Subsection \ref{subsect-delivery-user}.
 Let us consider the $(K,L,M,N)$ multiaccess  coded caching   problem,
where $t=\frac{KM}{N} \in \left\{0,1, \ldots,  \left\lfloor
\frac{K}{L} \right\rfloor  \right\}$. In other words, the whole library is totally cached $t$ times in the system.
Define that $K'=K-t(L-1)$.\footnote{\label{footnote:theorem1}   From Lemmas \ref{lamm-2} and \ref{lamm-3}, when $K\leq tL+1$ the scheme with minimum load is obtained. So we only need to consider the case $K>tL+1$. Then the assumption $K'>0$ always holds when $K>tL+1$. } For this multiaccess  caching problem, we search
a $(K',F',Z,S)$ PDA $\mathbf{P}=(p_{j,k})_{j\in[F'], k\in [K']}$ for the shared-link caching system  where the whole library is also totally cached $t$ times,  i.e.,
$\frac{K'Z}{ F'} = \frac{KM}{N} =t$. We define
\begin{align}
 \mathcal{A}_{j}=\{k\ :\ p_{j,k}=*, k\in [K']\}, \ \forall j\in [F'] \label{eq-PDA-star-set}
\end{align}
as the column label set of $\mathbf{P}$ where the entries in $j^{\text{th}}$ row are stars. We have $|\mathcal{A}_{j}|=t$ because each packet is cached exactly $t$ times.  So PDA $\mathbf{P}$ should satisfy
\begin{itemize}
  \item C$4$. Each row of $\mathbf{P}$ has exactly $t=\frac{K'Z}{F'}$ stars.
  \end{itemize}
%From Line 4 of Algorithm \ref{alg:PDA}, it can be seen that $\mathcal{A}_{j}$ also represents the set of users in $[K']$ who cache the $j^{\text{th}}$ packet of each file. Assume that all the sets $\mathcal{A}_{j}$ have the same size, i.e., $|\mathcal{A}_{1}|=\ldots=|\mathcal{A}_{F}|=t$. We can check that $t=\frac{K'Z}{F}$ since each column has $Z$ stars. So $t=\frac{K'Z}{F}=\frac{KM}{N}$, i.e., $\frac{M}{N}=\frac{ZK'}{FK}$.

 %with $\frac{M}{N}=\frac{ZK'}{FK}$.

%For any positive integers $K$, $L$, $M$ and $N$ satisfying that if there exists a $(K',F,Z,S)$ PDA satisfying that each row has exactly $t$ stars.   First the following notations are useful.

\iffalse
   For any $g\in[K]$ and $j\in [F']$, define
\begin{equation}\label{eq-caching-index}
 \mathcal{C}^{(g)}_{j}=\hbox{Mod}\big(\{\mathcal{A}_{j}[h]+h(L-1) +g-2 \ :\ h\in[t]\}, K \big)+1 .
\end{equation}
From \eqref{eq-PDA-star-set} and \eqref{eq-caching-index}, we can design our placement strategy in Section~\ref{subsec-caching-nodes}, where each user in the    multiaccess  caching model can not get the same packet from different neighbouring cache-nodes.
\fi

For the $(K,L,M,N)$ multiaccess  coded caching   problem,
we divide each file $W_{n}$ where $n\in [N]$ into $K$ parts with equal length,  $W_{n}=\left(W^{(1)}_{n},\ldots,W^{(K)}_n \right)$. Denote the set of the $g^{th}$ parts  by $\mathcal{W}^{(g)}=\left\{W^{(g)}_{1},  \ldots, W^{(g)}_{N} \right\}$, for each $g\in [K]$.
 As shown in the sketch of the proof   in Section~\ref{sub:illustrated example},
we divide the whole caching procedure into $K$ separate rounds, where in the $g^{th}$ round we only deal with $\mathcal{W}^{(g)}$.
In the $g^{\text{th}}$ round where $g\in [K]$, our construction  contains three steps: the generations for the node-placement array   $\mathbf{C}^{(g)}$,   the user-retrieve array $\mathbf{U}^{(g)}$, and the user-delivery array $\mathbf{Q}^{(g)}$,   where the definitions of these three arrays are given in Definition~\ref{defn:three arrays}.

%{\red COMMENTS OF KAI: I STRONGLY SUGGEST TO FIRST PROVIDE A COMPLETE EXAMPLE TO SHOW HOW WE EXTEND THE MN PDA TO THE CENTRALIZED MULTI-ACCESS MODEL. THEN WE PROVIDE THE GENERAL PROOF. THE EXAMPLE PROVIDED AT THE BEGINNING IS USED TO GIVE THE READERS A ROUGH IDEA ON WHAT WE DO. IT IS EXTREMELY HARD TO UNDERSTAND IF WE DIRECTLY PROVIDE THE GENERAL DESCRIPTION.}

\subsection{Placement strategy for cache-nodes: Generation of $\mathbf{C}^{(g)}$ for $g\in [K]$}
\label{subsec-caching-nodes}
The main objective in this step is that each $L$ neighbouring cache-nodes do not cache any common packet.
Let us first consider the case where $g=1$.
 The detailed placement is as follows, which is based on a  $(K',F',Z,S)$ PDA $\mathbf{P}=(p_{j,k})_{j\in [F'],k\in[K']}$ for the shared-link caching system  with $K'=K-t(L-1)$ and $\frac{K' Z}{ F'} = \frac{KM}{N} $.
 %each user in the    multiaccess  caching model can not get the same packet from different neighbouring cache-nodes.
\iffalse
We divide each file $W_{n}$ where $n\in [N]$ into $K$ parts with equal length,  $W_{n}=\left(W^{(1)}_{n},,W^{(K)}_n \right)$.
Denote the $g^{th}$ part set by $\mathcal{W}^{(g)}=\left\{W^{(g)}_{1}, W^{(g)}_{2}, \ldots, W^{(g)}_{N} \right\}$ for each $g\in [K]$.
\fi

We divide
 each part $W^{(1)}_{n}$ where $n\in [N]$ into $F'$ equal-length packets,  $W^{(1)}_{n}=(W^{(1)}_{n,1}, \ldots,W^{(1)}_{n,F'})$. Then each cache-node $C_k$ where $k\in [K]$  caches the following set of packets (recall that $\text{Mod}(b,a)\in \{1,\ldots,a \}$),
\begin{align}
\label{eq-caching-node}
& \mathcal{Z}^{(1)}_{C_k}=\{W^{(1)}_{n,j}\ :\ k\in \mathcal{C}^{(1)}_{j}, j\in [F'], n\in [N] \}, \\
&\text{where }  \mathcal{C}^{(1)}_{j}=
%\hbox{Mod}\big(
\{\mathcal{A}_{j}[h]+h(L-1)    :  h\in[|\mathcal{A}_{j}|]\}
%, K \big)
.\label{eq-caching-index}
\end{align}
Notice that $\mathcal{C}^{(1)}_{j}$ represents the set of cache-nodes   which cache the $j^{\text{th}}$ packet of each part in $\mathcal{W}^{(1)}$.  As the meaning of stars in PDA, we   use the following $F'\times K$ node-placement array $\mathbf{C}^{(1)}=(c^{(1)}_{j,k})_{j\in [F'],k\in [K]}$, where
\begin{align}
\label{eq-array-node-caching}
c^{(1)}_{j, k}=\left\{
\begin{array}{cc}
* & \hbox{if}\ \  k\in \mathcal{C}^{(1)}_{j} \\
null & \hbox{otherwise}
\end{array},
\right.
\end{align}
%From \eqref{eq-caching-node} and \eqref{eq-array-node-caching}, we have that
and each entry at the position $(j,k)$ of $\mathbf{C}^{(1)}$ is star if and only if the packets $W^{(1)}_{n,j}$ for all $n\in [N]$ are cached by cache-node $C_k$.   %{\red COMMENTS OF KAI: HERE, WE NEED TO CLARIFY THE CONNECTION BETWEEN $(K',F,Z,S)$ PDA $\mathbf{P}=(p_{j,k})_{j\in [F],k\in[K']}$ AND $\mathbf{C}^{(g)}$. ARE THEY WITH THE SAME FORM AFTER SOME SHIFT? IT IS UNCLEAR TO SAY ``Based on a $(K',F,Z,S)$ PDA $\mathbf{P}=(p_{j,k})_{j\in [F],k\in[K']}$, WE GENERATE $\mathbf{C}^{(g)}$.''}
 By the above construction, for any two distinct integers $h$, $h'\in [t]$ and for any $j\in [F']$, $g\in [K]$, the following inequality
\begin{align}
|\mathcal{C}^{(1)}_j[h]-\mathcal{C}^{(1)}_j[h']|=|\mathcal{A}_j[h]-\mathcal{A}_j[h']+(h-h')(L-1)|\geq L \label{eq:key placement}
\end{align}
holds. Thus any $L$ neighbouring cache-nodes do not cache any common packet.

After determining $\mathbf{C}^{(1)}$, we can obtain $\mathbf{C}^{(g)}$ where $g\in [K]$ by simply  cyclically  right-shifting $\mathbf{C}^{(1)}$ by $g-1$ positions.
  Then we let
\begin{align*}
\mathbf{C}=\left[
              \mathbf{C}^{(1)};
              \mathbf{C}^{(2)};
              \cdots;
              \mathbf{C}^{(K)}
            \right]
\end{align*}
to represent the cached contents of the cache-nodes. %From \eqref{eq-caching-index} and \eqref{eq-array-node-caching},  $\mathbf{C}^{(g)}$ can be obtained by cyclically right-shifting $\mathbf{C}^{(1)}$ by $g-1$ positions.
Each column of $\mathbf{C}$ is the concatenation of all the columns of $\mathbf{C}^{(1)}$. In addition, by the constraint,  the total number of stars in $(K',F',Z,S)$ PDA equals the total number of stars in $\mathbf{C}^{(1)}$. Hence,  the number of stars in each column of $\mathbf{C}$ is $ZK'$. The total number of packets of each file equals the number of rows of $\mathbf{C}$, i.e., $KF'$. So the needed memory  size of each cache-node is
\begin{align*}
\frac{ZK'}{KF'}N=\frac{Z}{F'}\cdot \frac{K'}{K}N=M,
\end{align*}
satisfying the memory size constraint.

\begin{example}\rm
\label{exam-caching-node}
Let us return to the example in Section~\ref{sub:illustrated example} with
  $K=N=8$ and $L=3$, where the caching procedure is divided into $8$ rounds.

In the first round, we consider  $\mathcal{W}^{(1)}=\{W^{(1)}_1,\ldots, W^{(1)}_8 \}$.   Each part $W^{(1)}_n$ where $n \in [8]$ is divided  into $6$ packets, i.e., $W^{(1)}_{n}=\left(W^{(1)}_{n,1},W^{(1)}_{n,2},\ldots,W^{(1)}_{n,6}\right)$.
   By \eqref{eq-PDA-star-set} we have
\begin{align}
\label{eq-A_j}
\mathcal{A}_1=\{1,2\},\ \ \mathcal{A}_2=\{1,3\},\mathcal{A}_3=\{1,4\},\mathcal{A}_4=\{2,3\},\mathcal{A}_5=\{2,4\},\mathcal{A}_6=\{3,4\}.
\end{align}
By \eqref{eq-caching-index} we have
\begin{equation*}
\mathcal{C}^{(1)}_{1}=\{3,6\},\ \ \mathcal{C}^{(1)}_{2}=\{3,7\},\ \ \mathcal{C}^{(1)}_{3}=\{3,8\},\ \
\mathcal{C}^{(1)}_{4}=\{4,7\},\ \ \mathcal{C}^{(1)}_{5}=\{4,8\},\ \ \mathcal{C}^{(1)}_{6}=\{5,8\}.
\end{equation*}
Then each cache-node $C_k$, $k\in [8]$, caches the packets from $\mathcal{W}^{(1)}$ as follows by \eqref{eq-caching-node}.
\begin{align*}
\mathcal{Z}^{(1)}_{C_1}&=\mathcal{Z}^{(1)}_{C_2}=\emptyset; \\
\mathcal{Z}^{(1)}_{C_3}&=\{W^{(1)}_{n,1},\ W^{(1)}_{n,2},\ W^{(1)}_{n,3}\ :\  n\in [8] \}; \\
\mathcal{Z}^{(1)}_{C_4}&=\{W^{(1)}_{n,4},\ W^{(1)}_{n,5}\ :\  n\in [8] \}; \\
\mathcal{Z}^{(1)}_{C_5}&=\{W^{(1)}_{n,6}\ :\  n\in [8] \};\\
\mathcal{Z}^{(1)}_{C_6}&=\{W^{(1)}_{n,1}\ :\  n\in [8] \}; \\
\mathcal{Z}^{(1)}_{C_7}&=\{W^{(1)}_{n,2},\ W^{(1)}_{n,4}\ :\  n\in [8] \}; \\
\mathcal{Z}^{(1)}_{C_8}&=\{W^{(1)}_{n,3},\ W^{(1)}_{n,5},\ W^{(1)}_{n,6}\ :\  n\in [8] \}.
\end{align*}
By \eqref{eq-array-node-caching} the above packets cached by cache-nodes can be represented by the first array $\mathbf{C}^{(1)}$ in Table \ref{tab-caching-1-2-nodes}, which is exactly the array $\mathbf{C}^{(1)}$ in Fig. \ref{Fig-construct}.
To get the array $\mathbf{C}^{(2)}$, we right-shift   $\mathbf{C}^{(1)}$ by one position.
\begin{table}[!htbp]
\center
\caption{Node-placement arrays $\mathbf{C}^{(1)}$ and $\mathbf{C}^{(2)}$.
\label{tab-caching-1-2-nodes}}
\begin{tabular}{|c|cccccccc|}
&\multicolumn{8}{|c|}{Node-placement array $\mathbf{C}^{(1)}$ for $\mathcal{W}^{(1)}$} \\ \hline
$n\in[8]$         &  $C_1$&$C_2$	&$C_3$	&$C_4$	&$C_5$	&$C_6$	&$C_7$  &	$C_8$	\\ \hline
$W^{(1)}_{n,1}$	&		&		&	$*$	&		&		&	$*$	&		&		\\
$W^{(1)}_{n,2}$	&		&		&	$*$	&		&		&		&	$*$	&		\\
$W^{(1)}_{n,3}$	&		&		&	$*$	&		&		&		&		&	$*$	\\
$W^{(1)}_{n,4}$	&		&		&		&	$*$	&		&		&	$*$	&		\\
$W^{(1)}_{n,5}$	&		&		&		&	$*$	&		&		&		&	$*$	\\
$W^{(1)}_{n,6}$	&		&		&		&		&	$*$	&		&		&	$*$	\\ \hline
\end{tabular}\ \ \ \ \ \ \
\begin{tabular}{|c|cccccccc|}
&\multicolumn{8}{|c|}{Node-placement array $\mathbf{C}^{(2)}$ for $\mathcal{W}^{(2)}$} \\ \hline
$n\in[8]$         &  $C_1$&$C_2$	&$C_3$	&$C_4$	&$C_5$	&$C_6$	&$C_7$  &	$C_8$	\\ \hline
$W^{(2)}_{n,1}$	&		&		&		&	$*$	&		&		&	$*$	&		\\
$W^{(2)}_{n,2}$	&		&		&		&	$*$	&		&		&		&	$*$	\\
$W^{(2)}_{n,3}$	&	$*$	&		&		&	$*$	&		&		&		&		\\
$W^{(2)}_{n,4}$	&		&		&		&		&	$*$	&		&		&	$*$	\\
$W^{(2)}_{n,5}$	&	$*$	&		&		&		&	$*$	&		&		&		\\
$W^{(2)}_{n,6}$	&	$*$	&		&		&		&		&	$*$	&		&		\\ \hline
\end{tabular}
\end{table} %From Table \ref{tab-caching-1-2-nodes}, we can see that $\mathbf{C}^{(2)}$ can be obtained by cyclically right-shifting $\mathbf{C}^{(1)}$ by $1$ position, and the
The total number of stars in $\mathbf{C}^{(1)}$ is $12$, which is the total number of stars in the $(4,6,3,4)$ PDA $\mathbf{P}$ in \eqref{eq-PDA-6-4}. Then the number of stars in each column of $\mathbf{C}$ is $12$. Since $\mathbf{C}$ has $6\times 8=48$ rows, the memory ratio of each cache-node is $\frac{M}{N}=\frac{12}{48}=\frac{1}{4}$.
\hfill $\square$
\end{example}

\subsection{Packets retrievable to users: Generation of $\mathbf{U}^{(g)}$ for $g\in [K]$}
\label{subsect-packet-user}
Let us also start with $g=1$.
 Since each user is connected to $L$ neighbouring cache-nodes in a cyclic wrap-around fashion, and can retrieve   the cached contents in those cache-nodes. Hence,  each user $U_k$ where $k \in [K]$ can retrieve the packet $W^{(1)}_{n,j}$ for all $n\in [N]$ if and only if $k$ is an element of the following set
\begin{align}
\mathcal{U}^{(1)}_{j}&=\bigcup\limits_{h\in[t]}
%\hbox{Mod}(
\{\mathcal{C}^{(1)}_j[h]-(L-1),
\ldots,\mathcal{C}^{(1)}_j[h]\}
% ,K)
 \nonumber \\
&=\bigcup\limits_{h\in[t]}
%\hbox{Mod}(
\{\mathcal{A}_j[h]+(h-1)(L-1),
\ldots,\mathcal{A}_j[h]+h(L-1)\}
% ,K)
.
\label{eq-user-index}
\end{align}
Notice that  $\mathcal{U}^{(1)}_{j}$ is the set of users   who can retrieve the $j^{\text{th}}$ packet of each part in $\mathcal{W}^{(1)}$.  Then user $U_k$ can retrieve the following packets of $\mathcal{W}^{(1)}$,
\begin{equation}\label{eq-packet-index-user}
\mathcal{Z}^{(1)}_{U_k}=\{W^{(1)}_{n,j}\ :\ k\in \mathcal{U}^{(1)}_{j}, n\in [N] \}.
\end{equation}

From~\eqref{eq:key placement}, we showed that any neighbouring cache-nodes do not cache any common packet. Recall that each packet is stored by $t$ cache-nodes and each cache-node is connected to  $ L$ users. Hence, we have $|\mathcal{U}^{(1)}_{j}|=tL$ for each $j\in [F']$, which means that each  packet is retrievable by $tL $ users.
\iffalse
For any two distinct integers $h$, $h'\in [t]$ and for any $j\in [F]$, $g\in [K]$, the following inequality
$$|\mathcal{C}^{(1)}_j[h]-\mathcal{C}^{(1)}_j[h']|=|\mathcal{A}_j[h]-\mathcal{A}_j[h']+(h-h')(L-1)|\geq L$$
holds. Since $K=K'+t(L-1)$, we have
\begin{align*}
K-|\mathcal{C}^{(1)}_j[h]-\mathcal{C}^{(1)}_j[h']|&= K-|\mathcal{A}_j[h]-\mathcal{A}_j[h']+(h-h')(L-1)|\\
&\geq K-(|\mathcal{A}_j[h]-\mathcal{A}_j[h']|+|(h-h')(L-1)|)\\
&\geq K-(K'+(t-1)(L-1))= L.
\end{align*} Thus we have  $|\mathcal{U}^{(1)}_{j}|=tL$.
In other words, each packet is stored by $t$ cache-nodes and can be retrieved by $tL$ users; thus by the construction of the placement,
\fi

From \eqref{eq-packet-index-user}, we can  generate the $F'\times K$  user-retrieve  array $\mathbf{U}^{(1)}=(u^{(1)}_{j,k})_{j\in[F'],k\in [K]}$,   where
\begin{align}
\label{eq-array-user-caching}
u^{(1)}_{j, k}=\left\{
\begin{array}{cc}
* & \hbox{if}\ \ k\in \mathcal{U}^{(1)}_{j} \\
null & \hbox{otherwise}
\end{array}.
\right.
\end{align}
In other words, each entry at position $(j,k)$ of $\mathbf{U}^{(1)}$   is star if and only if the packets $W^{(1)}_{n,j}$ for all $n\in [N]$ can be retrieved by user $U_k$.

\iffalse
From \eqref{eq-array-user-caching}, $\mathbf{U}^{(1)}$ has the following properties.
\begin{remark}\rm
\label{remark-1}Similar to the discussion of packet and cache-node array $\mathbf{C}^{(1)}$ defined in \eqref{eq-array-node-caching}, we can also have that
\begin{itemize}
\item each entry in the $j^{th}$ and $k^{th}$ column of $\mathbf{U}^{(1)}$,  is star if and only if the packets $W^{(1)}_{n,j}$ for all $n\in [N]$ can be retrieved by user $U_k$; otherwise, the entry is null.
\item $\mathbf{U}^{(g)}$ can be obtained by cyclically right-shifting $\mathbf{U}^{(1)}$ by $g-1$ positions.
\end{itemize}
\hfill $\square$
\end{remark}
\fi

\begin{remark}\rm
\label{remark-2}
The number of null entries in $j^{th}$ row of $\mathbf{U}^{(1)}$ equals the number of integer entries in the $j^{th}$ row of PDA $\mathbf{P}$ since the number of null entries in each $j^{th}$ row is $K-tL=K-t(L-1)-t=K'-t$.
\hfill $\square$
\end{remark}

After determining $\mathbf{U}^{(1)}$, we can obtain $\mathbf{U}^{(g)}$ where $g\in [K]$ by simply  cyclically  right-shifting $\mathbf{U}^{(1)}$ by $g-1$ positions.

\begin{example}
\label{exam-packets-user-array}
\rm
Let us return to the example in Section~\ref{sub:illustrated example}. From \eqref{eq-user-index}, we have
\begin{align*}
\mathcal{U}^{(1)}_{1}=\{1,2,3,4,5,6\},\ \ \mathcal{U}^{(1)}_{2}=\{1,2,3,5,6,7\},\ \
\mathcal{U}^{(1)}_{3}=\{1,2,3,6,7,8\},\\
\mathcal{U}^{(1)}_{4}=\{2,3,4,5,6,7\},\ \
\mathcal{U}^{(1)}_{5}=\{2,3,4,6,7,8\},\ \ \mathcal{U}^{(1)}_{6}=\{3,4,5,6,7,8\}.
\end{align*}
From \eqref{eq-packet-index-user}, the users can retrieve  the packets from $\mathcal{W}^{(1)}$ as follows,
\begin{align*}
\mathcal{Z}^{(1)}_{U_1}&=\{W^{(1)}_{n,1},\ W^{(1)}_{n,2},\ W^{(1)}_{n,3}\ :\  n\in [8] \}; \\
\mathcal{Z}^{(1)}_{U_2}&=\{W^{(1)}_{n,1},\ W^{(1)}_{n,2},\ W^{(1)}_{n,3},\ W^{(1)}_{n,4},\ W^{(1)}_{n,5}\ :\  n\in [8] \}; \\
\mathcal{Z}^{(1)}_{U_3}&=\{W^{(1)}_{n,1},\ W^{(1)}_{n,2},\ W^{(1)}_{n,3},\ W^{(1)}_{n,4},\ W^{(1)}_{n,5},\ W^{(1)}_{n,6}\ :\  n\in [8] \}; \\
\mathcal{Z}^{(1)}_{U_4}&=\{W^{(1)}_{n,1},\ W^{(1)}_{n,4},\ W^{(1)}_{n,5},\ W^{(1)}_{n,6}\ :\  n\in [8] \}; \\
\mathcal{Z}^{(1)}_{U_5}&=\{W^{(1)}_{n,1},\ W^{(1)}_{n,2},\ W^{(1)}_{n,4},\ W^{(1)}_{n,6}\ :\  n\in [8] \}; \\
\mathcal{Z}^{(1)}_{U_6}&=\{W^{(1)}_{n,1},\ W^{(1)}_{n,2},\ W^{(1)}_{n,3},\ W^{(1)}_{n,4},\ W^{(1)}_{n,5},\ W^{(1)}_{n,6}\ :\  n\in [8] \}; \\
\mathcal{Z}^{(1)}_{U_7}&=\{W^{(1)}_{n,2},\ W^{(1)}_{n,3},\ W^{(1)}_{n,4},\ W^{(1)}_{n,5},\ W^{(1)}_{n,6}\ :\  n\in [8] \}; \\
\mathcal{Z}^{(1)}_{U_8}&=\{W^{(1)}_{n,3},\ W^{(1)}_{n,5},\ W^{(1)}_{n,6}\ :\  n\in [8] \}.
\end{align*}
Hence, we can generate the user-retrieve array $\mathbf{U}^{(1)}$ in Table \ref{tab-caching-1-2-users}, which is   exactly the array $\mathbf{U}^{(1)}$ in Fig. \ref{Fig-construct}.
To get the array $\mathbf{U}^{(2)}$, we right-shift   $\mathbf{U}^{(1)}$ by one position.
It can be also seen that the number of null entries in each row of $\mathbf{U}^{(1)}$ equals the number of integer entries in each row of the $(4,6,3,4)$ PDA $\mathbf{P}$.
\begin{table}[!htbp]
\center
\caption{User-retrieve arrays $\mathbf{U}^{(1)}$ and $\mathbf{U}^{(2)}$.\label{tab-caching-1-2-users}}
\begin{tabular}{|c|cccccccc|}
&\multicolumn{8}{|c|}{User-retrieve array $\mathbf{U}^{(1)}$ for $\mathcal{W}^{(1)}$} \\ \hline
$n\in[8]$         &  $U_1$&$U_2$	&$U_3$	&$U_4$	&$U_5$	&$U_6$	&$U_7$  &	$U_8$	\\ \hline
$W^{(1)}_{n,1}$	&	$*$	&	$*$	&	$*$	&	$*$	&	$*$	&	$*$	&		&		\\
$W^{(1)}_{n,2}$	&	$*$	&	$*$	&	$*$	&		&	$*$	&	$*$	&	$*$	&		\\
$W^{(1)}_{n,3}$	&	$*$	&	$*$	&	$*$	&		&		&	$*$	&	$*$	&	$*$	\\
$W^{(1)}_{n,4}$	&		&	$*$	&	$*$	&	$*$	&	$*$	&	$*$	&	$*$	&		\\
$W^{(1)}_{n,5}$	&		&	$*$	&	$*$	&	$*$	&		&	$*$	&	$*$	&	$*$	\\
$W^{(1)}_{n,6}$	&		&		&	$*$	&	$*$	&	$*$	&	$*$	&	$*$	&	$*$	\\ \hline
\end{tabular}\ \ \ \ \ \ \
\begin{tabular}{|c|cccccccc|}
&\multicolumn{8}{|c|}{User-retrieve array $\mathbf{U}^{(2)}$ for $\mathcal{W}^{(2)}$} \\ \hline
$n\in[8]$         &  $U_1$&$U_2$	&$U_3$	&$U_4$	&$U_5$	&$U_6$	&$U_7$  &	$U_8$	\\ \hline
$W^{(2)}_{n,1}$	&		&	$*$	&	$*$	&	$*$	&	$*$	&	$*$	&	$*$	&		\\
$W^{(2)}_{n,2}$	&		&	$*$	&	$*$	&	$*$	&		&	$*$	&	$*$	&	$*$	\\
$W^{(2)}_{n,3}$	&	$*$	&	$*$	&	$*$	&	$*$	&		&		&	$*$	&	$*$	\\
$W^{(2)}_{n,4}$	&		&		&	$*$	&	$*$	&	$*$	&	$*$	&	$*$	&	$*$	\\
$W^{(2)}_{n,5}$	&	$*$	&		&	$*$	&	$*$	&	$*$	&		&	$*$	&	$*$	\\
$W^{(2)}_{n,6}$	&	$*$	&		&		&	$*$	&	$*$	&	$*$	&	$*$	&	$*$	\\ \hline
\end{tabular}
\end{table}
\hfill $\square$
\end{example}

\subsection{Delivery strategy: Generation of $\mathbf{Q}^{(g)}$ for $g\in [K]$}
\label{subsect-delivery-user}
We also start with $g=1$. By first letting $\mathbf{Q}^{(1)}=\mathbf{U}^{(1)}$, the next step is to fill the nulls in $\mathbf{Q}^{(1)}$ by some integers.
  From Remark \ref{remark-2}, for each $j\in [F']$ we can get a new array $\mathbf{Q}^{(1)}$ as follows. We
replace  the $h^{th}$ (where $h\in [K'-t]$)  null entry in the $j^{th}$ row of $\mathbf{U}^{(1)}$ by  the $h^{th}$  integer entry in the $j^{th}$ row of $\mathbf{P}$.
  Specifically, for each $j\in [F']$, define
\begin{align}
\label{eq-integer-PDA}
\overline{\mathcal{A}}_j=\{k\ :\ p_{j,k}\in [S],k\in[K']\},
\end{align}
i.e., the column label set of $\mathbf{P}$ where the entries in $j^{th}$ row are integers, and
\begin{align}
\label{eq-ineteger-array}
\overline{\mathcal{U}}^{(1)}_j=\{k\ :\ p_{j,k}=null,k\in[K]\},
\end{align}
i.e., the column label set of $\mathbf{U}^{(1)}$ where the entries in $j^{th}$ row are nulls. From Remark \ref{remark-2}, we have
$$
|\overline{\mathcal{A}}_j|=|\overline{\mathcal{U}}_j|=K'-t.
$$
 We then define a one-to-one mapping $\psi_{j}$ from $\overline{\mathcal{U}}^{(1)}_j$ to $\overline{\mathcal{A}}_j$:
\begin{align*}
\psi_{j}(\overline{\mathcal{U}}^{(1)}_j [h])=\overline{\mathcal{A}}_j[h], \ \forall h\in [K'-t], j\in [F'].
\end{align*}
  Then the entry of the new array $\mathbf{Q}^{(1)}=(q^{(1)}_{j,k})_{j\in[F'],k\in[K]}$ can be written as follows.
\begin{align}
\label{eq-array-user-PDA}
q^{(1)}_{j, k}=\left\{
\begin{array}{cc}
s & \hbox{if}\ \ k\in\overline{\mathcal{U}^{(1)}_j}   \text{ and }  p_{j,\psi_{j}(k)}=s; \\
* & \hbox{otherwise.}
\end{array}
\right.
\end{align}
Hence, the alphabet set of the resulting $\mathbf{Q}^{(1)}$ consists of $[S]$ and symbol star.

After determining $\mathbf{Q}^{(1)}$, the delivery procedure in the first round is as follows.
  From Line 9 of Algorithm~\ref{alg:PDA}, for each integer $s\in [S]$ and  $g=1$, the server sends the following multicast message
  \begin{align}
  \bigoplus_{q^{(1)}_{j,k}=s,j\in[F'],k\in [K]}W^{(1)}_{d_{U_k},j},
  \end{align}
if $\mathbf{Q}^{(1)}$ defined in \eqref{eq-array-user-PDA} satisfies the condition C$3$ in Definition \ref{def-PDA}, such that each user can directly decode its required packet from each multicast message transmitted by the server.
% if the following requirement is satisfied,
% such that each user can directly decode its required packet from each multicast message transmitted by the server. \\
%{\bf Requirement:} $\mathbf{Q}^{(1)}$ defined in \eqref{eq-array-user-PDA} satisfies the condition C$3$ in Definition \ref{def-PDA}.\\[0.2cm]

Now let us discuss the conditions on $\mathbf{P}$ to guarantee that the $\mathbf{Q}^{(1)}$ defined in \eqref{eq-array-user-PDA} satisfies the condition C$3$ in Definition \ref{def-PDA}. For any integer $s\in [S]$, assume that there exist two different entries, say, $q^{(1)}_{j_1,k_1}=q^{(1)}_{j_2,k_2}=s$. Since the integers in each row of the $(K',F',Z,S)$ PDA $\mathbf{P}$ are different, then all the integers in each row of $\mathbf{Q}^{(1)}$ must be different by \eqref{eq-array-user-PDA}. Then we have $j_1\neq j_2$. From \eqref{eq-ineteger-array} there exist two unique integers $h_1\in [K'-t]$ and $h_2\in [K'-t]$ satisfying $k_1=\overline{\mathcal{U}}^{(1)}_{j_1}[h_1]$ and $k_2=\overline{\mathcal{U}}^{(1)}_{j_2}[h_2]$. So we have $\psi_{j_1}(k_1)=\overline{\mathcal{A}}_{j_1}[h_1]$ and $\psi_{j_2}(k_2)=\overline{\mathcal{A}}_{j_2}[h_2]$. Let $k'_1=\psi_{j_1}(k_1)$ and $k'_2=\psi_{j_2}(k_2)$. From the definition of mapping $\psi_{j_1}$ and $\psi_{j_2}$, we have $p_{j_1,k'_1}=p_{j_2,k'_2}=s$. Let $(\mathcal{A}_{j_1}\bigcup \{k'_1\})[i_1]=k'_1$ and $(\mathcal{A}_{j_2}\bigcup \{k'_2\})[i_2]=k'_2$ for some integers $i_1$ and $i_2$. By \eqref{eq-user-index}, \eqref{eq-array-user-caching} and \eqref{eq-array-user-PDA}, we have $k_1=k'_1+(i_1-1)(L-1)$ and $k_2=k'_2+(i_2-1)(L-1)$. By \eqref{eq-array-user-caching}, $q^{(1)}_{j_1,k_2}=q^{(1)}_{j_2,k_1}=*$ if and only if $k_1\in \mathcal{U}^{(1)}_{j_2}$ and $k_2\in \mathcal{U}^{(1)}_{j_1}$. So
$\mathbf{Q}^{(1)}$ defined in \eqref{eq-array-user-PDA} satisfies the condition C$3$ if and only if $\mathbf{P}$ satisfies the following condition.
\begin{itemize}
\item C$5$. For any two distinct entries $p_{j_1,k'_1}$ and $p_{j_2,k'_2}$, if $p_{j_1,k'_1}=p_{j_2,k'_2}\in[S]$, we have
    \begin{eqnarray*}
  k'_1+(i_1-1)(L-1)\in \mathcal{U}^{(1)}_{j_2}\ \ \text{and}\ \ k'_2+(i_2-1)(L-1)\in \mathcal{U}^{(1)}_{j_1}
    \end{eqnarray*} hold where $k'_1=(\mathcal{A}_{j_1}\bigcup\{k'_1\})[i_1]$ and $k'_2=(\mathcal{A}_{j_2}\bigcup\{k'_2\})[i_2]$ for some integers $i_1$, $i_2\in[t+1]$.
\end{itemize}

Since the server totally sends $SK$ multicast messages of packets in all rounds and each file is divided into $KF'$ packets, the achieved load is
 $R=\frac{SK}{KF'}=\frac{S}{K}$. So the following result can be obtained.
 \begin{proposition}\rm
\label{pro-fundament}
 If a $(K',F',Z,S)$ PDA $\mathbf{P}$ satisfies C$4$ and C$5$, there exists a $(K=K'+t(L-1),L,M,N)$ multiaccess coded caching scheme with the memory ratio $\frac{M}{N}=\frac{K' Z}{KF'}$, subpacketization $F=KF'$ and load $R=\frac{S}{F'}$.
 \hfill $\square$
\end{proposition}
From Proposition \ref{pro-fundament}, we can use the proposed transformation approach to extend a PDA satisfying C$4$ and C$5$ to the multiaccess caching problem. Let us then consider the MN PDA, and we have the following lemma whose proof can be found in Appendix~\ref{sec:proof of MN lemma}.
\begin{lemma}
\label{lemma-integer-PDA}
\rm
 The MN PDA satisfies C$4$ and C$5$.
\hfill $\square$
\end{lemma}
Let $\mathbf{P}$ be a $\left(K',{K'\choose t},{K'-1\choose t-1}, {K'\choose t+1} \right)$ MN PDA.
From Proposition \ref{pro-fundament}, we can get a multiaccess coded caching scheme for the $(K,L,M,N)$ multiaccess  caching problem, with the subpacketization $K{K'\choose t}=K{K-t(L-1)\choose t}$ and load $R_1=\frac{{K'\choose t+1}K}{K{K'\choose t}}=\frac{K'-t}{t+1}=\frac{K-tL}{t+1}$.
 Hence, we proved Theorem~\ref{th-1}.

\begin{example}
\label{exam-delivery-array}
\rm
Let us return to the example in Section~\ref{sub:illustrated example} again. From \eqref{eq-integer-PDA} and \eqref{eq-ineteger-array} we have
\begin{align*}
\overline{\mathcal{A}}_1=\{ 3,4\}, \ \ \ \overline{\mathcal{A}}_2=\{2,4\}, \ \ \ \overline{\mathcal{A}}_3=\{2,3\}, \ \ \ \overline{\mathcal{A}}_4=\{ 1,4\}, \ \ \ \overline{\mathcal{A}}_5=\{1,3\}, \ \ \ \overline{\mathcal{A}}_6=\{1,2\},
\end{align*}
and
\begin{align*}
\overline{\mathcal{U}}^{(1)}_1=\{7,8\}, \ \ \ \overline{\mathcal{U}}^{(1)}_2=\{4,8\}, \ \ \ \overline{\mathcal{U}}^{(1)}_3=\{4,5\}, \ \ \ \overline{\mathcal{U}}^{(1)}_4=\{1,8\}, \ \ \ \overline{\mathcal{U}}^{(1)}_5=\{1,5\}, \ \ \ \overline{\mathcal{U}}^{(1)}_6=\{1,2\},
\end{align*}
respectively. Then the following mappings can be obtained.
\begin{align*}
\begin{array}{cccccc}
  \psi_1(7)=3, & \psi_1(8)=4,& \psi_2(4)=2, & \psi_2(8)=4, & \psi_3(4)=2, & \psi_3(5)=3, \\
  \psi_4(1)=1, & \psi_4(8)=4, & \psi_5(1)=1, & \psi_5(5)=3, & \psi_6(1)=1, & \psi_6(2)=2.
\end{array}
\end{align*}
From \eqref{eq-array-user-PDA}, the user-delivery array  $\mathbf{Q}^{(1)}$  can be obtained in Table \ref{tab-delivery-1-2-users},
which is exactly the array $\mathbf{Q}^{(1)}$ in Fig. \ref{Fig-construct}.  To get the array $\mathbf{Q}^{(2)}$, we right-shift   $\mathbf{Q}^{(1)}$ by one position.
\begin{table}[!htbp]
\center
\caption{User-delivery arrays $\mathbf{Q}^{(1)}$ and $\mathbf{Q}^{(2)}$.\label{tab-delivery-1-2-users}}
\begin{tabular}{|c|cccccccc|}
&\multicolumn{8}{|c|}{User-delivery array $\mathbf{Q}^{(1)}$ for $\mathcal{W}^{(1)}$} \\ \hline
$n\in[8]$         &  $1$&$2$	&$3$	&$4$	&$5$	&$6$	&$7$  &	$8$	\\ \hline
$W^{(1)}_{n,1}$	&	$*$	&	$*$	&	$*$	&	$*$	&	$*$	&	$*$	&	$1$	&	$2$	\\
$W^{(1)}_{n,2}$	&	$*$	&	$*$	&	$*$	&	$1$	&	$*$	&	$*$	&	$*$	&	$3$	\\
$W^{(1)}_{n,3}$	&	$*$	&	$*$	&	$*$	&	$2$	&	$3$	&	$*$	&	$*$	&	$*$	\\
$W^{(1)}_{n,4}$	&	$1$	&	$*$	&	$*$	&	$*$	&	$*$	&	$*$	&	$*$	&	$4$	\\
$W^{(1)}_{n,5}$	&	$2$	&	$*$	&	$*$	&	$*$	&	$4$	&	$*$	&	$*$	&	$*$	\\
$W^{(1)}_{n,6}$	&	$3$	&	$4$	&	$*$	&	$*$	&	$*$	&	$*$	&	$*$	&	$*$	\\ \hline
\end{tabular}\ \ \ \ \ \ \
\begin{tabular}{|c|cccccccc|}
&\multicolumn{8}{|c|}{User-delivery array $\mathbf{Q}^{(2)}$ for $\mathcal{W}^{(2)}$} \\ \hline
$n\in[8]$         &  $1$&$2$	&$3$	&$4$	&$5$	&$6$	&$7$  &	$8$	\\ \hline
$W^{(2)}_{n,1}$	&	$2$	&	$*$	&	$*$	&	$*$	&	$*$	&	$*$	&	$*$	&	$1$	\\
$W^{(2)}_{n,2}$	&	$3$	&	$*$	&	$*$	&	$*$	&	$1$	&	$*$	&	$*$	&	$*$	\\
$W^{(2)}_{n,3}$	&	$*$	&	$*$	&	$*$	&	$*$	&	$2$	&	$3$	&	$*$	&	$*$	\\
$W^{(2)}_{n,4}$	&	$4$	&	$1$	&	$*$	&	$*$	&	$*$	&	$*$	&	$*$	&	$*$	\\
$W^{(2)}_{n,5}$	&	$*$	&	$2$	&	$*$	&	$*$	&	$*$	&	$4$	&	$*$	&	$*$	\\
$W^{(2)}_{n,6}$	&	$*$	&	$3$	&	$4$	&	$*$	&	$*$	&	$*$	&	$*$	&	$*$	\\ \hline
\end{tabular}
\end{table}
 It can be  checked that both of  $\mathbf{Q}^{(1)}$ and $\mathbf{Q}^{(2)}$ satisfy the condition C$3$ of Definition \ref{def-PDA}. %Then the server sends the following multicast messages in the first two rounds,  listed in Table \ref{tab-delivery-1-2-users}.
%\begin{table}[!htbp]
%\center
%\caption{Multicast messages for $\mathcal{W}^{(1)}$ and $\mathcal{W}^{(2)}$ .\label{tab-transmission-1}}
%\begin{tabular}{c|c}
%Time slots & Multicast messages from $\mathcal{W}^{(1)}$\\ \hline
%$1$ & $W^{(1)}_{1,4}	\bigoplus	W^{(1)}_{4,2}	\bigoplus	W^{(1)}_{7,1}$\\
%$2$ & $W^{(1)}_{1,5}	\bigoplus	W^{(1)}_{4,3}	\bigoplus	W^{(1)}_{8,1}$\\
%$3$ & $W^{(1)}_{1,6}	\bigoplus	W^{(1)}_{5,3}	\bigoplus	W^{(1)}_{8,2}$\\
%$4$ & $W^{(1)}_{2,6}	\bigoplus	W^{(1)}_{5,5}	\bigoplus	W^{(1)}_{8,4}$\\ \hline
%\end{tabular} \ \ \ \ \ \
%\begin{tabular}{c|c}
%Time slots & Multicast message from $\mathcal{W}^{(2)}$\\ \hline
%$1$ & $W^{(2)}_{2,4}\bigoplus W^{(2)}_{5,2}	\bigoplus	W^{(2)}_{8,1}$\\
%$2$ & $W^{(2)}_{2,5}\bigoplus W^{(2)}_{5,3}	\bigoplus	W^{(2)}_{1,1}$\\
%$3$ & $W^{(2)}_{2,6}\bigoplus W^{(2)}_{6,3}	\bigoplus	W^{(2)}_{1,2}$\\
%$4$ & $W^{(2)}_{3,6}\bigoplus W^{(2)}_{6,5}	\bigoplus	W^{(2)}_{1,4}$\\ \hline
%\end{tabular}
%\end{table}
%We can easily check that each user can decode its required packets from the multicast messages in Table \ref{tab-transmission-1}.
Then the load is $R_1=\frac{4\times 8}{6\times 8}=\frac{2}{3}$.
\hfill $\square$
\end{example}

From Proposition \ref{pro-fundament}, we can use a PDA satisfying C$4$ and C$5$ with a low row number to obtain a multiaccess coded caching scheme with  a low subpacketization. In Appendix \ref{sec:proof of Partition theorem}, we show that the PDA in \cite{YCTC} satisfies  C$4$ and C$5$. Consequently Theorem~\ref{th-3} can be obtained. Similarly, we can also show that the PDAs in \cite{CJYT,SZG,YTCC} satisfy C$4$ and C$5$.

\section{Further Improved Transformation Approach}
\label{sec-further}
 In this section, we will show that based on the placement strategy proposed in Section~\ref{subsec-caching-nodes}, we can further reduce the transmission load of the scheme proposed in Section~\ref{subsect-delivery-user},  if some
additional conditions given in the following Proposition~\ref{pro-improve} are satisfied.

Given a $(K',F',Z,S)$ PDA, assume $t=|\mathcal{A}_1|=\cdots=|\mathcal{A}_F'|$. For each $h\in [t]$, let
$a_h=\min\{\mathcal{A}_j[h] : j\in[F']\}$ and $b_h=\max\{\mathcal{A}_j[h] : j\in[F']\}$.
\begin{proposition}\rm
\label{pro-improve}
Given a $(K',F',Z,S)$ PDA $\mathbf{P}$ satisfying C$4$ and C$5$, if $\lambda_h=a_h+L-b_h>0$ holds for each $h\in [t ]$, then for any positive integers $L$, $M$ and $N$, there exists a $(K=K'+t(L-1),L,M,N)$ multiaccess coded caching scheme with the memory ratio $\frac{M}{N}=\frac{K'Z}{KF'}$, subpacketization $F=KF'$ and transmission load $R=\frac{S}{F'}\cdot \frac{K-\sum_{h=1}^{t }\lambda_h}{K}$.
\hfill $\square$
\end{proposition}
\begin{proof}
 From Proposition \ref{pro-fundament}, we have a $(K=K'+t(L-1),L,M,N)$ multiaccess coded caching scheme with the memory ratio $\frac{M}{N}=\frac{K' Z}{KF'}$, subpacketization $F=KF'$ and load $R=\frac{S}{F'}$. Now let us consider the transmission load by further compressing the  multicast messages.

For any integer $j\in [F']$, we know that
$$\mathcal{U}^{(1)}_{j}=\bigcup_{h\in[t]}\left\{\mathcal{A}_{j}[h]+(h-1)(L-1), \ldots,\mathcal{A}_{j}[h]+h(L-1)\right\}$$
is the set of users in $[K]$ who can retrieve the packets $W^{(1)}_{n,j}$ for all $n\in[N]$. So for any two different positive integers $j_1$, $j_2\in[F']$, a user $U_k$, $k\in [K]$ can retrieve the packets $W^{(1)}_{n,j_1}$ and $W^{(1)}_{n,j_2}$ for all $n\in[N]$ if and only if
$k\in \mathcal{U}^{(1)}_{j_1}\bigcap \mathcal{U}^{(1)}_{j_2}$.
\begin{subequations}
For each integer $h\in[t]$,
\begin{align}
&\bigcap\limits_{j\in[F']}\{\mathcal{A}_{j}[h]+(h-1)(L-1), \ldots,\mathcal{A}_{j}[h]+h(L-1)\} \label{eq:first term intersection} \\
&= \{a_h+(h-1)(L-1), \ldots,a_h+h(L-1)\}\bigcap \{b_h+(h-1)(L-1), \ldots,b_h+h(L-1)\} \label{eq:subseteq second} \\
&= \{b_h+(h-1)(L-1), \ldots,a_h+h(L-1)\}, \label{eq:imply that}
\end{align}
\end{subequations}
always holds where~\eqref{eq:subseteq second} comes from that $a_h=\min\{\mathcal{A}_j[h] : j\in[F']\}$ and $b_h=\max\{\mathcal{A}_j[h] : j\in[F']\}$. From \eqref{eq:first term intersection} and \eqref{eq:imply that}, each user from $\{b_h+(h-1)(L-1), \ldots,a_h+h(L-1)\}$ can retrieve all the packets of the first part. From~\eqref{eq:imply that}, it can be seen that there are exactly
\begin{align*}
\sum_{h=1}^{t }|\{b_h+(h-1)(L-1), \ldots,a_h+h(L-1)\}|=\sum_{h=1}^{t }\lambda_h
\end{align*}users who can re-construct all the $S$ multicast messages for $\mathbf{Q}^{(1)}$ from their retrieved cache-nodes. By the symmetry, considering all $\mathbf{Q}^{(g)}$ where $g\in [K]$,
among all the $KS$ multicast messages in the delivery phase, each user can re-construct $\sum_{h=1}^{t }\lambda_h S$ multicast messages. Hence, we can transmit $(K-\sum_{h=1}^{t }\lambda_h)S$ random linear combinations of the $KS$ multicast messages. \footnote{\label{footnote:RLC} Instead of random linear combinations, we can also use the parity check matrix of Minnimum Distance Seperable (MDS) code or Cauchy matrix as in~\cite{Lint}, to encode the $KS$ multicast messages. In each of these matrices whose dimension is dimension $m_1\times m_2$ where $m_1\leq m_2$, every $m_1$ columns are linearly independent. }
Then the transmission load is
\begin{eqnarray*}
R=\frac{(K-\sum_{h=1}^{t }\lambda_h)S}{KF'}=\frac{S}{F'}\cdot\frac{K-\sum_{h=1}^{t }\lambda_h}{K}.
\end{eqnarray*}
Then the proof is completed.
\end{proof}
 \begin{example}\rm
\label{exam-reducing}
Let us   return to the example in Section~\ref{sub:illustrated example} with $K=N=8$ and $L=3$, where the caching procedure is divided into $8$ rounds. Given the $(4,6,3,4)$ PDA listed in Fig. \ref{Fig-construct}, from \eqref{eq-A_j} we have $a_1=1$, $b_1=3$, $a_2=2$, $b_2=4$ and then $\lambda_1=\lambda_2=1>0$. From Proposition \ref{pro-improve} there exists a $(K,L,M,N)=(8,3,2,8)$ multiaccess coded caching scheme with the subpacketization $F=KF'=48$ and transmission load $R=\frac{S}{F'}\cdot \frac{K-\sum_{h=1}^{t }\lambda_h}{K}=\frac{4}{6}\cdot\frac{8-2}{8}=\frac{1}{2}$.  This transmission load is lower than the transmission load  of the scheme in Section~\ref{sub:illustrated example} (i.e., $\frac{2}{3}$).
\end{example}

By \eqref{Eqn_Def_AN}, in MN PDA we have $|\mathcal{A}_{1}|=\ldots=\mathcal{A}_{{K'\choose t}}=t$ and $a_{h}=h$, $b_{h}=K'-(t-h)$ for each $h\in [t]$. For any positive integer $L$, if
\begin{align*}
t+1<K'<t+L
\end{align*}
which implies $tL+1<K<tL+L$ in Theorem \ref{th-1}, we have
\begin{align*}
\lambda_h=a_h+L-b_h=h+L-(K'-(t-h))=t+L-K'>0.
\end{align*}
 From Theorem \ref{th-1} and Proposition \ref{pro-improve}, the following result can be directly obtained.
 \begin{corollary}\rm
\label{cor-1}
For the $(K,L,M,N)$ centralized multiaccess coded caching problem, if $M=\frac{N t}{K}$ where $t\in \left\{0,1,\ldots, \left\lfloor
\frac{K}{L} \right\rfloor \right\}$ and $tL+1<K<tL+L$,
the following load is achievable,
\begin{align}
R_{3}= \frac{K-tL}{t+1} \frac{(t+1)(K-tL)}{K}=\frac{(K-tL)^2}{K}. \label{eq:R2}
\end{align}
\end{corollary}

We can check that the transmission load in Corollary \ref{cor-1} is exactly the transmission load in Lemma \ref{lamm-2}. This implies that when $tL+1<K<tL+L$, the scheme in \cite{RK} has lower transmission load. As shown in~\eqref{eq:comparison to RK}, when $(t+1)L<K$ the proposed scheme in Theorem \ref{th-1} has strictly lower transmission load than the scheme in \cite{RK}.

Now let us consider the scheme in Theorem \ref{th-3}. We can check that in the PDA  from Appendix \ref{sec:proof of Partition theorem},  $|\mathcal{A}_{1}|=\ldots=\mathcal{A}_{q^{m-1}}=m$ and $b_h-a_{h}=q-1$ for each $h\in [m]$. For any positive integer $L$, if
$q\leq L$ in Theorem \ref{th-3}, we have
\begin{align*}
\lambda_h=L-q+1>0.
\end{align*}
From Theorem \ref{th-3} and Proposition \ref{pro-improve}, we can obtain the following result.
 \begin{corollary}\rm
\label{cor-2}
For any positive integers $m$, $q\geq 2$ and $L$ with $L\geq q$, there exits an $(m(q+L-1),L,M,N)$   multiaccess coded caching scheme with $\frac{M}{N}=\frac{1}{q+L-1}$, the subpacketization $F=m(q+L-1)q^{m-1}$ and transmission load $R_{4}= \frac{2(q-1)^2}{q+L-1}$.
\end{corollary}
\section{Conclusion}\label{sec-conclusion}
In this paper,  we consider the multiaccess coded caching problem and propose a novel transformation approach to extend any shared-link PDA satisfying two conditions to the considered problem. The resulting scheme has the maximum local caching gain and the same coded caching gain as  the original PDA. By applying  our transformation approach into the MN PDA, the delivery scheme was proved to be approximately optimal  when $K$ is sufficiently large, under the constraint of the used placement. Finally we also provided  an improved transformation approach to further reduced the load.

On-going works include the generalization of the proposed transformation approach to more general multiaccess topologies, such as the line multiaccess topology where users are connected to different numbers of  cache-nodes,   the planar  multiaccess topology, and the multiaccess coded caching problem with  distance-dependent retrieval costs.

\appendices
\section{Proof of Theorem~\ref{thm:converse}}
\label{sec:converse proof}
For any integer $q \in [tL:K]$, we define that
\begin{align}
\mathcal{S}_{q}= \{\mathcal{T} \subseteq [K-q+1:K] : |\mathcal{T}|=t, \text{Mod}(j_1-j_2,K)\geq L, \ j_1,j_2 \in \mathcal{T} \text{ and } j_1 \neq j_2  \}.\label{eq:def of Sq}
\end{align}
It can be computed that
\begin{align}
|\mathcal{S}_{q}|=\binom{q-X}{t-1}\frac{q}{t}, \label{eq:length of Sq}
\end{align}
 where $X:=t L-t+1$.
Thus under the placement in~\eqref{eq:our subfile division} for each $\mathcal{T} \in \mathcal{S}_{K}$ and $i\in [N]$, there is a subfile $W_{i,\mathcal{T}}$ of $W_i$ cached by the cache-nodes with indices in $\mathcal{T}$. Each file has $F$ packets and  each subfile has $\frac{F t}{\binom{K-tL+t-1}{t-1} K}$ packets.

Consider the demand vector $\mathbf{d}=(1,\ldots,K)$.  For each $k\in [K]$,
 each subfile in $\cup_{\mathcal{T} \in \mathcal{S}_{K-k-L+1}}W_{k, \mathcal{T}} $ is demanded by user $k$ and cannot be retrieved by the users in $[k]$.
By the converse bound in~\cite{WTP,YMA}, we have
\begin{subequations}
\begin{align}
&R^{\star}_1  F \geq \sum_{k\in [K]} \sum_{\mathcal{T} \in \mathcal{S}_{K-k-L+1}} |W_{q, \mathcal{T}}| \label{eq:acyclic converse}   \\
&  =  \frac{F t}{\binom{K-tL+t-1}{t-1} K} \sum_{k\in [K]}  \binom{K-k-L+1-X}{t-1}\frac{K-k-L+1}{t} \label{eq:acyclic converse step 2}\\
& = \frac{F }{\binom{K-tL+t-1}{t-1} K} \sum_{k\in [K]} (K-k-L+1 ) \binom{K-k-L+1-X}{t-1} \\
&= \frac{F }{\binom{K-tL+t-1}{t-1} K} \left( \binom{K-L-X}{t-1}+(K-L-1) \binom{K-L-X+1}{t} - \binom{K-L-X}{t+1} \right), \label{eq:from pascal}
\end{align}
  \end{subequations}
where~\eqref{eq:acyclic converse step 2} comes from~\eqref{eq:length of Sq} and~\eqref{eq:from pascal} comes from the Pascal's triangle. Hence, we prove~\eqref{eq:converse bound}.

Next, we focus on the case where $K \gg tL$. In this case, we have
\begin{subequations}
\begin{align}
& \frac{ R_{1} }{R^{\star}_1 } \leq \frac{(K-t L)\binom{K-tL+t-1}{t-1} K }{(t+1) \left( \binom{K-L-X}{t-1}+(K-L-1) \binom{K-L-X+1}{t} - \binom{K-L-X}{t+1} \right)} \\
&\leq \frac{(K-t L)\binom{K-tL+t-1}{t-1} K }{(t+1) \left(  (K-L-1) \binom{K-L-X+1}{t} - \binom{K-L-X}{t+1} \right)}\\
&\leq \frac{(K-t L)\binom{K-tL+t-1}{t-1} K }{(t+1) \left(    (K-L-1) \binom{K-L-X+1}{t} -   \frac{K-L-1}{t+1} \binom{K-L-X+1}{t} \right)}\\
& = \frac{(K-t L)\binom{K-tL+t-1}{t-1} K }{ t    (K-L-1) \binom{K-L-X+1}{t}  } \\
& = \frac{K   (K-t L) }{(K-L-1)  (K-L-X+1)}\cdot \frac{K-tL+t-1}{K-L-X}\cdot \frac{K-tL+t-2}{K-L-X-1}\cdots \frac{K-tL+1}{K-L-X-t+2}
%\frac{K \times (K-t L) \times (K-tL+t-1) \times  (K-tL+t-2)  \times  \cdots  \times  (K-tL+1) }{(K-L-1) \times (K-L-X+1) \times (K-L-X)  \times  \cdots  \times (K-L-X-t+2)}
\\
&\leq \frac{K   (K-t L) }{(K-L-1)  (K-L-X+1)} \left( \frac{K-tL+1}{K-L-X-t+2} \right)^{t-1} \\
&= \frac{K   (K-t L) }{(K-L-1)  (K-L-tL+t)} \left( 1 + \frac{L}{K-L-tL+1} \right)^{t-1} \\
&\approx \frac{K   (K-t L) }{(K-L-1)  (K-L-tL+t)}  \left( 1 + \frac{L(t-1)}{K-L-tL+1} \right) \label{eq:binom app}\\
&\approx 1,
\end{align}
  \end{subequations}
where~\eqref{eq:binom app} comes from the binomial approximation. Hence, we prove~\eqref{eq:approximate opt}.
\section{Proof of~\eqref{eq:comparison to HKD},~\eqref{eq:comparison to RK} and~\eqref{eq:comparison to SR}}
\label{sec:proof of cor1}
\subsection{Proof of~\eqref{eq:comparison to HKD}}
\label{sec:proof of HKD}
We focus on the case where $L $ does not divide $K$.
When $M_1=\frac{N}{K} \lfloor\frac{K}{2L}\rfloor$,  we have
\begin{align}
\frac{R_{\text{HKD}}}{R_1}&= \frac{\frac{K-\lfloor\frac{K}{2L}\rfloor}{\lfloor\frac{K}{2L}\rfloor+1}\cdot\frac{\frac{K}{L}-t}
{\frac{K}{L}-\lfloor\frac{K}{2L}\rfloor}}
{\frac{K-tL}{t+1}}
=\frac{K-\lfloor\frac{K}{2L}\rfloor}{K-L\lfloor\frac{K}{2L}\rfloor}
\cdot\frac{t+1}{\lfloor\frac{K}{2L}\rfloor+1} > \frac{K-\lfloor\frac{K}{2L}\rfloor}{K-L\lfloor\frac{K}{2L}\rfloor} .\label{eq:larger than 2}
\end{align}
Let us then focus on a memory size $M_1 \leq M$. The achieved load by the scheme in Lemma~\ref{lamm-3} is obtained by memory-sharing between the memory-load tradeoff points $\left( \frac{N}{K} \lfloor\frac{K}{2L}\rfloor, R_{\text{HKD}} \right)$ and $\left(\frac{N}{L},0 \right)$.
In addition,   the achieved load by the proposed scheme is no worse than the load obtained by   memory-sharing between the memory-load tradeoff points $\left(\frac{N}{K} \lfloor\frac{K}{2L} \rfloor, R_{1} \right)$ and $\left(\frac{N}{L},0 \right)$.
Additionally with~\eqref{eq:larger than 2}, we can prove that  with $M$, the multiplicative gap between the achieved loads by  Lemma~\ref{lamm-3} and  by the proposed scheme is larger than $\frac{K-\lfloor\frac{K}{2L}\rfloor}{K-L\lfloor\frac{K}{2L}\rfloor}$, which coincides with~\eqref{eq:comparison to HKD}.
\subsection{Proof of~\eqref{eq:comparison to RK}}
\label{sec:proof of RK}
 Let us focus on the case where $K>(t+1)L$.  From~\eqref{eq:lem2 load} and~\eqref{eq:load R1}, we have
\begin{align*}
\frac{R_1}{R_{\text{RK}}}&=\frac{K-tL}{t+1}\cdot \frac{K}{(K-tL)^2}
 =\frac{1}{(t+1)(1- \frac{tL}{K})}  < \frac{1}{(t+1)(1- \frac{t }{t+1})} =1.
\end{align*}
\subsection{Proof of~\eqref{eq:comparison to SR}}
\label{sec:proof of SR}
 From  Lemma \ref{lamm-3}, it can be seen that
\begin{align*}
R_{\text{SR}}\geq \sum_{h=\frac{K-tL+2}{2}}^{K-tL}\frac{2}{1+\lceil \frac{tL}{h}\rceil}.
\end{align*}
We focus on the non-trivial corner points at the memory sizes $M=\frac{N t}{K}$ where $t \in \left\{ 1,\ldots,  \left\lfloor
\frac{K}{L} \right\rfloor  \right\}$.
We first show that when $\frac{KM}{2N}\left(1-\frac{ML}{N}\right)\geq  1$, we have $R_1 < R_{\text{SR}}$. More precisely, we have
\begin{align}
& \sum_{h=\frac{K-tL+2}{2}}^{K-tL}\frac{2}{1+\lceil \frac{tL}{h}\rceil}\geq\sum_{h=\frac{K-tL+2}{2}}^{K-tL}\frac{2}{2+ \frac{tL}{h}} \nonumber \\
&= \overbrace{\underbrace{\frac{1}{2+ \frac{tL}{\frac{K-tL+2}{2}}}+\frac{1}{2+ \frac{tL}{\frac{K-tL+2}{2}}}}_{2 \text{ terms}}+\underbrace{\frac{1}{2+ \frac{tL}{\frac{K-tL+2}{2}+1}}+
\frac{1}{2+ \frac{tL}{\frac{K-tL+2}{2}+1}}}_{2 \text{ terms}}+\ldots+
\underbrace{\frac{1}{2+ \frac{tL}{K-tL}}+\frac{1}{2+ \frac{tL}{K-tL}}}_{2 \text{ terms}}}^{K-tL \text{ terms}} \nonumber \\
&> (K-tL)\frac{1}{2+ \frac{tL}{\frac{K-tL+2}{2}}}. \label{eq:step 1 R1 better than SR}
\end{align}
From~\eqref{eq:step 1 R1 better than SR}, it can be seen that
\begin{align*}
&\frac{R_{\text{SR}}}{R_{1}}>\frac{(K-tL)\frac{1}{2+ \frac{tL}{\frac{K-tL+2}{2}}}}{\frac{K-tL}{t+1}} =\frac{t+1}{2+ \frac{tL}{\frac{K-tL+2}{2}}}
=\frac{\frac{KM}{N}+1}{2}\cdot\frac{K-\frac{KML}{N}+2}{K+2}\\
&=\frac{\frac{KM}{N}+1}{2}\cdot\left(1-\frac{KML}{N(K+2)}\right)  >\frac{\frac{KM}{N}+1}{2}\cdot\left(1-\frac{ML}{N}\right) \\
&>\frac{KM}{2N}\left(1-\frac{ML}{N}\right) \geq 1.
%\end{split}
\end{align*}
Hence, we showed that when   $\frac{KM}{2N}\left(1-\frac{ML}{N}\right)\geq  1$, it holds that $R_1 < R_{\text{SR}}$.

Let us then consider the regime where $\frac{KM}{2N}\left(1-\frac{ML}{N}\right)< 1$. In this case, we have either
$\frac{M}{N}>\frac{1}{2L}+\sqrt{\frac{1}{4L^2}-\frac{2}{KL}}$ or  $\frac{M}{N}<\frac{1}{2L}-\sqrt{\frac{1}{4L^2}-\frac{2}{KL}}$. When $K \gg L$, we have
\begin{align*}
&\frac{1}{2L}+\sqrt{\frac{1}{4L^2}-\frac{2}{KL}} \to \frac{1}{L}; \\
&\text{and } \frac{1}{2L}-\sqrt{\frac{1}{4L^2}-\frac{2}{KL}} \to 0.
\end{align*}
Recall that $M \leq \frac{N}{L}$. Hence, we can prove that when  $K \gg L$, the regime $\frac{KM}{2N}\left(1-\frac{ML}{N}\right)< 1$ does not exist.

In conclusion,  if $K \gg L$,  we can prove that
$R_1 <R_{\text{SR}}$ for $M=\frac{N t}{K}$ where $t \in \left\{ 1,\ldots,  \left\lfloor
\frac{K}{L} \right\rfloor  \right\}$.
\section{Proof of Lemma~\ref{lemma-integer-PDA}}
\label{sec:proof of MN lemma}
 We can easily check that the MN PDA $\mathbf{P}$ satisfies C$4$ by the construction of the MN PDA in~\eqref{Eqn_Def_AN}. We will then prove that C$5$ holds. For any integer $s\in [S]$, assume that there exist two different integers $j_1$, $j_2\in [F']$ and two different integers $k'_1$, $k'_2\in [K']$, satisfying $p_{j_1,k'_1}=p_{j_2,k'_2}=s$. From the construction of MN PDA in \eqref{Eqn_Def_AN} we have $\mathcal{S}=\mathcal{A}_{j_1}\cup\{k'_1\}=\mathcal{A}_{j_2}\cup\{k'_2\}$ and $k'_1\not\in\mathcal{A}_{j_1}$, $k'_2\not\in\mathcal{A}_{j_2}$. This implies that $k'_1\in \mathcal{A}_{j_2}$ and $k'_2\in \mathcal{A}_{j_1}$. Let
$$\mathcal{A}_{j_2}[h'_2]=k'_1,\ \  \mathcal{A}_{j_1}[h'_1]=k'_2,\ \ k'_1=\mathcal{S}[i_1],\ \ k'_2=\mathcal{S}[i_2],$$
for some positive integers $h'_1$, $h'_2\in[t]$ and  $i_1$, $i_2\in[t+1]$.
When $k'_2<k'_1$, we have $i_2=h'_1$ and $i_1=h'_2+1$. Then we have
\begin{eqnarray*}
&&k'_1+(i_1-1)(L-1)=\mathcal{A}_{j_2}[h'_2]+h'_2(L-1)\\
%&=\{h'_2+h_1+(h'_2-1)(L-1),\ldots,h'_2+h_1+h'_2(L-1)\} \nonumber \\
%&=\{\overline{\mathcal{A}}_{j_1}[h_1]+(h'_2-1)(L-1),\ldots,k_1+h'_2(L-1)\} \nonumber \\
&\in&\{\mathcal{A}_{j_2}[h'_2]+(h'_2-1)(L-1),\ldots,\mathcal{A}_{j_2}[h'_2]+h'_2(L-1)\} \\
&\subseteq&\mathcal{U}^{(1)}_{j_2};
\end{eqnarray*}
\begin{eqnarray*}&& k'_2+(i_2-1)(L-1)=\mathcal{A}_{j_1}[h'_1]+(h'_1-1)(L-1)\\%
%&=\{(h'_1-1)+h_2+(h'_1-1)(L-1),\ldots,(h'_1-1)+h_2+h'_1(L-1)\} \nonumber \\
%&=\{k_2+(h'_1-1)(L-1),\ldots,k_2+h'_1(L-1)\} \nonumber \\
&\in&\{\mathcal{A}_{j_1}[h'_1]+(h'_1-1)(L-1),\ldots,\mathcal{A}_{j_1}[h'_1]+h'_1(L-1)\} \\
&\subseteq&\mathcal{U}^{(1)}_{j_1}.
\end{eqnarray*} So C$5$ holds. Similarly when $k'_2>k'_1$, we can also show $k'_1+(i_1-1)(L-1)\in\mathcal{U}^{(1)}_{j_2}$ and $k'_2+(i_2-1)(L-1)\in\mathcal{U}^{(1)}_{j_1}$. Then the proof is completed.

\section{Proof of Theorem~\ref{th-3}}
\label{sec:proof of Partition theorem}
By setting  $t=1$ in~\cite[Construction 1]{CWZW}, we have the following shared-link PDA construction.
\begin{construction}\rm(\cite[Construction 1]{CWZW})
\label{construction}
For any positive integers $m$ and $q\geq 2$, let
\begin{eqnarray}
\label{eq-generator-s=m-1}
\begin{split}
\mathcal{F'}&=\left\{\left(f_{1},f_{2},\ldots,f_{m-1}, f_m\right)\ :\ f_1,f_2,\ldots,f_{m-1}\in [q],f_m=\hbox{Mod}\left(\sum_{i=1}^{m-1} f_i,q\right)\right\},\\
\mathcal{K}'&=\{(\delta,b)\ :\ \delta\in [m], b\in [q]\}.
\end{split}\end{eqnarray}
Then an  $|\mathcal{F'}|\times |\mathcal{K}'|$ array $\mathbf{P}=(p_{{\bf f},(\delta,b)})$, ${\bf f}\in \mathcal{F'}$, $(\delta,b)\in \mathcal{K}'$, can be defined in the following way
\begin{eqnarray}
\label{eq-constr-PDA}
p_{{\bf f},(\delta,b)}=\left\{
\begin{array}{ll}
({\bf e},n_{\bf e}) & \textrm{if}~f_{\delta}\neq b\\
* & \textrm{otherwise}
\end{array}
\right.
\end{eqnarray}
where ${\bf e}=(e_1,e_{2},\ldots,e_{m})\in[q]^m$ such that \begin{eqnarray}
\label{eq-putting integer}
e_i=\left\{
\begin{array}{ll}
b & \textrm{if}\ i=\delta\\[0.2cm]
f_i & \textrm{otherwise} \end{array}
\right.
\end{eqnarray} and $n_{\bf e}$ is the occurrence order of vector ${\bf e}$ occurs in column $(\delta,b)$.
\hfill $\square$
\end{construction}
The following PDA proposed in~\cite[Theorem 6]{CWZW}  can be directly obtained from the above construction.
\begin{lemma}\rm(\cite[Theorem 6]{CWZW})
\label{lem-property-integer}
The array $\mathbf{P}$ generated by Construction \ref{construction} is an $(mq,q^{m-1},q^{m-2},(q-1)q^{m-1})$ PDA. Furthermore,
\begin{itemize}
\item Each row has exactly $m$ stars;%
%\item If there are two distinct entries being the same vector $({\bf e},n_{{\bf e}})$, say $p_{{\bf f},(\delta-1)q+b}=p_{{\bf f}',(\delta'-1)q+b'}=({\bf e},n_{{\bf e}})$, then $\delta\neq \delta'$;
\item For each row ${\bf f}\in \mathcal{F'}$ and each $\delta \in[m]$, there is exactly one star and $q-1$ vector entries in the columns $(\delta,b)$ for all $b\in[q]$.
\end{itemize}
\hfill $\square$
\end{lemma}
It is worth noting that the above $\mathbf{P}$ was first constructed in~\cite[Construction A]{YCTC}. However the rule of defining integer entries in Construction 1 in \cite{CWZW} is much simpler than that of Construction A in \cite{YCTC}. In order to prove our claim easier, we use the notions in \cite{CWZW}.

From the first property in Lemma \ref{lem-property-integer}, $\mathbf{P}$ satisfies C$4$ in Proposition \ref{pro-fundament}. Since each vector $(\delta,b)$, $\delta\in [m]$ and $b\in[q]$, can be uniquely represented by an integer $k=(\delta-1)q+b$. So for convenience of our proof, we will not distinguish to vector $(\delta,b)$ and the integer $(\delta-1)q+b$ in the following. According to \eqref{eq-constr-PDA}, the column label set of $\mathbf{P}$, where the entries in ${\bf f}^{\text{th}}$ row, ${\bf f}\in\mathcal{F}$, are stars, defined in \eqref{eq-PDA-star-set} can be written as
\begin{eqnarray*}
\mathcal{A}_{{\bf f}}&=&\{(\delta-1) q+f_{\delta}:\ \delta\in [m]\}.
\end{eqnarray*}
Then the set $\mathcal{U}^{(1)}_{{\bf f}}$ defined in \eqref{eq-user-index} can be written as follows.
\begin{eqnarray}
\label{eq-partition-U}
%\mathcal{C}^{(1)}_{{\bf f}}&=&\{\delta L +(\delta-1)(q-1)+f_{\delta}-1:\ \delta\in [m]\}\\
\mathcal{U}^{(1)}_{{\bf f}}=\{(\delta -1)L +(\delta-1)(q-1)+f_{\delta},\ldots,\delta L +(\delta-1)(q-1)+f_{\delta}-1:\ \delta\in [m]\}.
\end{eqnarray}

Now let us consider C$5$ in Proposition \ref{pro-fundament}. For any two different vectors ${\bf f}$, ${\bf f}'\in \mathcal{F}'$ and two different integers $k'_1$, $k'_2\in \mathcal{K}'$, assume that $p_{{\bf f},k'_1}=p_{{\bf f}',k'_2}$ are not star entries. Let
$$k'_1=(\delta_1-1)q+b_1,\ \ k'_2=(\delta_2-1)q+b_2,\ \ k'_1=(\mathcal{A}_{{\bf f}}\cup \{k'_1\})[i_1],\ \ k'_2=(\mathcal{A}_{{\bf f}'}\cup \{k'_2\})[i_2],$$
for some integers $\delta_1$, $\delta_2\in[m]$ and $b_1$, $b_2\in[q]$ and $i_1$, $i_2\in [m+1]$. From the second property in Lemma \ref{lem-property-integer}, until $k'_1$ we have
\begin{eqnarray*}\label{eq-k'_1-star}
\left\{\begin{array}{cc}
i_1=\delta_1+1\ \ \ \ \ & \hbox{if } b_1>f_{\delta_1}\\[0.2cm]
i_1=\delta_1 \ \ \ \ \ \ \ \ \ \ &\hbox{if } b_1<f_{\delta_1}
\end{array}\right.
\end{eqnarray*}
star entries and then
\begin{eqnarray*}\label{eq-k_1-value}
k'_1+(i_1-1)(L-1)=\left\{\begin{array}{cc}
\delta_1 L+(\delta_1-1)(q-1)+b_1-1\ \ \ & \hbox{if } b_1>f_{\delta_1}\\[0.2cm]
(\delta_1-1)L+(\delta_1-1)(q-1)+b_1\ \ &\hbox{if } b_1<f_{\delta_1}
\end{array}.\right.
\end{eqnarray*}
Similarly we can get
\begin{eqnarray*}\label{eq-k_2-value}
k'_2+(i_2-1)(L-1)=\left\{\begin{array}{cc}
\delta_2L+(\delta_2-1)(q-1)+b_2-1\ \ \ & \hbox{if } b_2>f'_{\delta_2}\\[0.2cm]
(\delta_2-1)L+(\delta_2-1)(q-1)+b_2\ \ &\hbox{if } b_2<f_{\delta_2}
\end{array}.\right.
\end{eqnarray*}Since $\mathbf{P}$ is a PDA, $p_{{\bf f},k'_2}=p_{{\bf f}',k'_1}=*$ always holds by C$3$ in Definition \ref{def-PDA}. Then $b_1=f'_{\delta_1}$ and $b_2=f_{\delta_2}$ can be derived by \eqref{eq-constr-PDA}. By \eqref{eq-partition-U} we have
\begin{eqnarray*}
&&k'_1+(i_1-1)(L-1)\\
&\in&\{(\delta_1 -1)L+(\delta_1 -1)(q-1)+b ,\delta_1 L+(\delta_1 -1)(q-1)+b -1\}\\
&=&\{(\delta_1 -1)L+(\delta_1 -1)(q-1)+f'_{\delta_1 },\delta_1 L+(\delta_1 -1)(q-1)+f'_{\delta_1 }-1\}\\
&\subseteq&\{(\delta_1-1)L +(\delta_1-1)(q-1)+f'_{\delta_1},\ldots,\delta_1  L +(\delta_1 -1)(q-1)+f'_{\delta_1}-1\}\subseteq \mathcal{U}^{(1)}_{{\bf f}'};
\end{eqnarray*}
\begin{eqnarray*}
&&k'_2+(i_2-1)(L-1)\\
&\in&\{(\delta_2-1)L+(\delta_2-1)(q-1)+b_2,\delta_2L+(\delta_2-1)(q-1)+b_2-1\}\\
&=&\{(\delta_2-1)L+(\delta_2-1)(q-1)+f_{\delta_2},\delta_2L+(\delta_2-1)(q-1)+f_{\delta_2}-1\}\\
&\subseteq&\{(\delta_2 -1)L +(\delta_2-1)(q-1)+f_{\delta_2},\ldots,\delta_2 L +(\delta_2-1)(q-1)+f_{\delta_2}-1\}\subseteq \mathcal{U}^{(1)}_{{\bf f} }.
\end{eqnarray*}So C$5$ holds. From Proposition \ref{pro-fundament}, our claim can be obtained. Then the proof is completed.


\begin{thebibliography}{1}
\bibitem{BBD}
E. Bastug, M. Bennis, and M. Debbah, Living on the edge: The role of proactive caching in 5G wireless networks, IEEE Commun. Magazine, vol. 52, pp. 82-89, Aug. 2014.

\bibitem{MN}
M. A. Maddah-Ali and U. Niesen, Fundamental Limits of Caching, {\it IEEE Trans. Inform. Theory}, vol. 60, no. 5, pp. 2856-2867, Mar. 2014.

\bibitem{WTP}
K. Wan, D. Tuninetti, and P. Piantanida, An Index Coding Approach to Caching With Uncoded Cache Placement, {\it IEEE Trans. Inform. Theory}, vol. 66, no. 3, pp. 1318-1332, Mar. 2020.

\bibitem{YMA}
Q. Yu, M. A. Maddah-Ali, and A. S. Avestimehr, The Exact Rate-Memory Tradeoff for Caching With Uncoded Prefetching, {\it IEEE Trans. Inform. Theory}, vol. 64, no. 2, pp. 1281-1296, Feb. 2018.

%\bibitem{ASK}
%S. Agrawal, K. V. S. Sree, P. Krishnan, Coded caching based on combinatorial designs, in Proc. {\it IEEE ISIT}, Paris,  7-12 July 2019, pp. 1227-1231.

%\bibitem{BLLX}
%Li B, Rui L L, Qiu X, et al. Content Caching Strategy for Edge and Cloud Cooperation Computing, {\it 2019 15th International Wireless Communications \& Mobile Computing Conference (IWCMC)}. IEEE, 2019: 260-265.

%\bibitem{GBKO}
%G. Barish and K. Obraczke, World wide web caching: Trends and techniques, {\it IEEE Commun. Mag}, vol. 38, no. 5, pp. 178-184, 2000.
\bibitem{yufactor2TIT2018}
Q. Yu, M. A. Maddah-Ali, and A. S. Avestimehr, Characterizing the Rate-Memory Tradeoff in Cache Networks Within a Factor of 2, {\it IEEE Trans. Inform. Theory}, vol. 65, no. 1, pp. 647-663, Jan. 2019.


\bibitem{SJTLD}
K. Shanmugam, M. Ji, A. M. Tulino, J. Llorca, and A. G. Dimakis, Finite-Length Analysis of Caching-Aided Coded Multicasting, {\it IEEE Trans. Inform. Theory}, vol.~62, no.~10, pp. 5524-5537, Oct. 2016.

\bibitem{YCTC}
Q. Yan, M. Cheng, X. Tang, and Q. Chen, On the Placement Delivery Array Design for Centralized Coded Caching Scheme, {\it IEEE Trans. Inform. Theory}, vol.~63, no.~9, pp. 5821-5833, Sep. 2017.

\bibitem{CJTY}
M. Cheng, J. Jiang,  X. Tang, and Q. Yan, Some Variant of Known Coded Caching Schemes With Good Performance, {\it IEEE Trans. Commun.}, vol. 68, no.3, pp. 1370-1377, Mar. 2020.

\bibitem{CJWY}
M. Cheng, J. Jiang, Q. Wang, Y. Yao, A Generalized Grouping Scheme in Coded Caching, {\it IEEE Trans. Commun.}, vol. 67, no. 5, pp. 3422-3430, May 2019.

\bibitem{CJYT}
M. Cheng, J. Jiang, Q. Yan, X.Tang, Constructions of Coded Caching Schemes With Flexible Memory Size, {\it IEEE Trans. Commun.}, vol. 67, no. 6, pp. 4166-4176, Jun. 2019.

\bibitem{MW}
J. Michel and Q. Wang, Placement Delivery Arrays From Combinations of Strong Edge Colorings, {\it IEEE Trans. Commun}, vol.68, no.10, pp. 5953-5964, Oct. 2020.

\bibitem{SZG}
C. Shangguan, Y. Zhang, and G. Ge, Centralized Coded Caching Schemes: A Hypergraph Theoretical Approach,  {\it IEEE Trans. Inform. Theory}, vol. 64, no. 8, pp. 5755-5766, Aug. 2018.

\bibitem{YTCC}
Q. Yan, X. Tang, Q. Chen, and M. Cheng, Placement Delivery Array Design Through Strong Edge Coloring of Bipartite Graphs, {\it IEEE Commun. Lett.}, vol. 22, no. 2, pp. 236-239, Feb. 2018.

\bibitem{ZCJ}
X. Zhong, M. Cheng, and J. Jiang, Placement Delivery Array Based on Concatenating Construction, {\it IEEE Communi. Letters}, vol. 24, no. 6, pp. 1216-1220, Jun. 2020.

\bibitem{LiuWirelesscaching}
D. Liu, B. Chen, C. Yang,  A. F. Molisch,  Caching at the wireless edge: design aspects, challenges, and future directions, {\it IEEE  Commun. Magazine}, vol. 54, no. 9, pp. 22-28, Sep. 2016.

\bibitem{FemtoCaching}
K. Shanmugam, N. Golrezaei, A. G. Dimakis, A. F. Molisch, and G. Caire, FemtoCaching: Wireless Content Delivery Through Distributed Caching Helpers, {\it IEEE Trans. Inform. Theory}, vol. 59, no. 12, pp. 8402-8413, Dec. 2013.

\bibitem{JHNS}
J. Hachem, N. Karamchandani, and S. N. Diggavi, Coded Caching for Multi-level Popularity and Access, {\it IEEE Trans. Inform. Theory}, vol. 63, no. 5, pp. 3108-3141, May 2017.

\bibitem{RK}
K. S. Reddy, N. Karamchandani, Rate-Memory Trade-off for Multi-Access Coded Caching With Uncoded Placement, {\it IEEE Trans. Commun}, vol. 68, no. 6, pp. 3261-3274, Jun. 2020.

\bibitem{SR}
S. Sasi, B. S. Rajan, An Improved Multi-access Coded Caching with Uncoded Placement, arXiv:2009.05377, Sep. 2020

 \bibitem{SPE}
 B. Serbetci, E. Parrinello and P. Elia, Multi-access coded caching: gains beyond cache-redundancy, in Proc. {\it IEEE Information Theory Workshop (ITW)},  Visby, Sweden, 2019, pp. 1-5.

\bibitem{MARB}
A. A. Mahesh, B. S. Rajan, Coded Caching Scheme with Linear Sub-packetization and its Application to Multi-Access Coded Caching, arXiv:2009.10923, Sep. 2020.

\bibitem{Structureindexcoding}
  K. S. Reddy and N. Karamchandani,  Structured index coding problem and multi-access coded caching,  arXiv preprint arXiv:2012.04705, Dec. 2020.

 \bibitem{OG}
E. Ozfatura and D. Gündüz, Mobility-Aware Coded Storage and Delivery, {\it IEEE Trans. Commun}, vol. 68, no. 6, pp. 3275-3285, Jun. 2020.


%\bibitem{FAGA}
%F. Cicirelli, A. Guerrieri, G. Spezzano, A. Vinci, O. Briante, A. Iera, and G. Ruggeri, Edge computing and social internet of things for large-scale smart environments development, {\it IEEE Internet of Things Journal}, vol. PP, no. 99, pp. 1-1, 2017.

%\bibitem{CLTW}
%M. Cheng, J. Li and X. Tang, R. Wei, Linear coded caching scheme for centralized networks, arXiv: 1810.06017 [cs. IT]

%\bibitem{MN-1}
%M. A. Maddah-Ali and Urs Niesen, Decentralized coded caching attains order-optimal memory-rate tradeoff, {\it IEEE/ACM Trans. Netw} , vol. 23, no. 4, pp. 1029-1040, Aug. 2015.

%\bibitem{JCM}
%M. Ji, G. Caire, and A. F. Molisch, Fundamental limits of caching in wireless D2D networks, {\it IEEE Trans. Inform. Theory}, vol. 62, no. 2, pp. 849-869, 2016.

%\bibitem{KNMD}
%N. Karamchandani, U. Niesen, M. A. Maddah-Ali, and S. Diggavi, Hierarchical coded caching, in Proc. {\it IEEE ISIT}, Honolulu, HI, Jun. 2014, pp. 2142-2146.



%\bibitem{RKEA}
%R. H. Katz, E. A. Brewer, The Case for Wireless Overlay Networks, {\it Mobile Computing}, Springer US, 1996.

%\bibitem{K}
%P. Krishnan, Coded caching via line graphs of bipartite graphs, in Proc. {\it IEEE ITW}, Guangzhou, Nov. 2018.

\bibitem{Lint}
J. H. van Lint, Introduction to Coding Theory, third version, Springer, 1999, Printed in Germany.

%\bibitem{JPJM}
%J. Pan and J. Mcelhannon, Future edge cloud and edge computing for internet of things applications, {\it IEEE Internet of Things Journal}, vol. PP, no. 99, pp. 1-1, 2017.

%\bibitem{RPMAN}
%R. Pedarsani, M. A. Maddah-Ali and U. Niesen, Online Coded Caching, {\it IEEE/ACM Transactions on Networking}, vol. 24, no. 2, pp. 836-845, April 2016.

%\bibitem{STDT}
%K. Shanmugam, A. G. Dimakis, J. Llorca and A. M. Tulino, A unified Ruzsa-Szem\'{e}redi graphs framework for finite-length coded caching, in Proc. 51st ACSSC, 631-635, 2017.



%\bibitem{SDLT}
%K. Shanmugam, A. G. Dimakis, J. Llorca, and A. M. Tulino, A unified Ruzsa-Szemer\'{e}di framework for finite-length coded caching, in Proc. {\it The 51st ACSSC}, Pacific Grove, CA, 2017, pp. 631-635.

%\bibitem{STD}
%K. Shanmugam, A. M. Tulino, and A. G. Dimakis, Coded Caching with Linear Subpacketization is Possible using Ruzsa-Szem\'{e}redi Graphs, in Proc. {\it IEEE ISIT}, Aachen, Germany, Jun. 2017, pp. 1237-1241.

%\bibitem{TR}
%L. Tang, A. Ramamoorthy, Coded caching schemes with reduced subpacketization from linear block codes, {\it IEEE Trans. Inform. Theory}, vol. 64, no. 4, pp. 3099-3120, 2018.

%\bibitem{TC}
 %C. Tian and J. Chen, Caching and delivery via interference elimination, in Proc. {\it IEEE Trans. Inf. Theory}, vol. 64, no. 3, pp. 1548-1560, 2018.


%\bibitem{AG}
%M. M. Amiri and D. G\"{u}nd\"{u}z, Fundamental limits of caching: Improved delivery rate-cache capacity trade-off, {\it IEEE Trans. Commun.}, vol. 65, no. 2, pp. 806-815, 2016.

%\bibitem{GR}
%H. Ghasemi and A. Ramamoorthy, Improved lower bounds for coded caching, in Proc. {\it IEEE ISIT}, Hong Kong, Jun. 2015, pp. 1696-1700.

%\bibitem{YMA}
%Q. Yu, M. A. Maddah-Ali, and A. S. Avestimehr, The exact rate-memory tradeoff for caching with uncoded prefetching, {\it IEEE Trans. Inform. Theory}, vol. 64, no. 2, pp. 1281-1296, 2018.


%\bibitem{STC}
%A. Sengupta, R. Tandon, and T. C. Clancy, Improved approximation of storage-rate tradeoff for caching via new outer
%bounds, in Proc. {\it IEEE ISIT}, Hong Kong, Jun. 2015, pp. 1691-1695.

%\bibitem{CLWZC}
%M. Cheng, D. Liang, K. Wan, M. Zhang, and G. Caire, A novel transformation approach of shared-link coded caching schemes for multiaccess networks, arXiv:2012.04483v2, Jan 2021.


\bibitem{CWZW}
M. Cheng, J. Wang, X. Zhong, and Q. Wang, A unified framework for constructing centralized coded caching schemes, arXiv: 3768116, May 20.
\end{thebibliography}
\end{document}